\documentclass[10pt,journal,compsoc]{IEEEtran}

\usepackage[T1]{fontenc}
\usepackage[utf8]{inputenc}
\usepackage[T1]{fontenc}
\usepackage[english]{babel}
\usepackage{todonotes}
\usepackage{xcolor}

\usepackage{enumitem}
\usepackage{graphicx}
\usepackage{pgfplots}
\usepackage{caption}
\usepackage{mathtools}
\usepackage{colortbl}
\usepackage{placeins}
\usepackage{float}
\usepackage{showlabels}
\usepackage{multirow}
\usepackage{amsmath,amsthm}
\usepackage{ragged2e}
\usepackage{listings}
\usepackage{svg}
\usepackage{babel}
\usepackage[backend=bibtex, style=numeric,]{biblatex}
\usepackage{hyperref}
\usepackage[position=bottom]{subfig}
\usepackage{algorithm} 
\usepackage{algpseudocode}
\usepackage[normalem]{ulem}
\usepackage{mathtools}

\DeclarePairedDelimiter\floor{\lfloor}{\rfloor}
\usepackage{soul}
\setuldepth{Berlin}

\usepgfplotslibrary{units}
\usepgfplotslibrary{ternary}
\date{}
\pgfplotsset{select coords between index/.style 2 args={
    x filter/.code={
        \ifnum\coordindex<#1\fi
        \ifnum\coordindex>#2\fi
    }
}}
\colorlet{colorLight1Blue}{blue!10}
\colorlet{colorLight1Green}{green!10}
\colorlet{colorLight1Orange}{orange!10}

\colorlet{colorLight3Blue}{blue!30}
\colorlet{colorLight3Green}{green!30}
\colorlet{colorLight3Orange}{orange!30}

\colorlet{colorLight5Blue}{blue!50}
\colorlet{colorLight5Green}{green!50}
\colorlet{colorLight5Orange}{orange!50}

\colorlet{colorLight7Blue}{blue!70}
\colorlet{colorLight7Green}{green!70}
\colorlet{colorLight7Orange}{orange!70}

\colorlet{colorLight9Blue}{blue!90}
\colorlet{colorLight9Green}{green!90}
\colorlet{colorLight9Orange}{orange!90}

\newtheorem{claim}{Claim}
\newtheorem{corollary}{Corollary}
\newtheorem{definition}{Definition}

\newcommand\Myperm[2][^n]{\prescript{#1\mkern-2.5mu}{}P_{#2}}

\begin{document}
\title{\textbf{SISSLE In Consensus-Based Ripple}: \textbf{S}ome \textbf{I}mprovements In \textbf{S}peed, \textbf{S}ecurity and \textbf{L}ast Mil\textbf{E} Connectivity}

\author{Mayank Mundhra,
        Chester Rebeiro
\IEEEcompsocitemizethanks{\IEEEcompsocthanksitem M. Mundhra is from Indian Institute of Technology Madras, India.\protect\\
Email: mayank.mundhra.2012@gmail.com
\IEEEcompsocthanksitem C. Rebeiro is a faculty at Indian Institute of Technology Madras,  India.\protect\\
Email: chester@cse.iitm.ac.in}

\thanks{Manuscript received April 10, 2021; revised April 10, 2021.}}

% \author{\IEEEauthorblockN{Mayank Mundhra}
% \IEEEauthorblockA{Indian Institute of Technology Madras, Chennai, India
% \\mayank.mundhra.2012@gmail.com} 
% \and
% \IEEEauthorblockN{Chester Rebeiro}The duration $y$ for Sub-round C of Candidate Set Generation phase (as discussed earlier in the paper) includes time for (1) min 3 hops\footnote{In this paper, we imply 'hops' to be vertex hops.} of IP in our overlay and (2) distributed nature of consensus (factor of 2). Thus $y \geq 6 \times x$.

% \IEEEauthorblockA{Indian Institute of Technology Madras, Chennai, India
% \\chester@cse.iitm.ac.in}}

% \markboth{Journal: IEEE/ACM Transactions on Networking,~Vol.~14, No.~8, April~2021}%
% {Mundhra \MakeLowercase{\textit{et al.}}: SISSLE In Consensus-Based Ripple: Some Improvements In Speed, Security and Last MilE Connectivity}

% \IEEEoverridecommandlockouts
% \makeatletter\def\@IEEEpubidpullup{6.5\baselineskip}\makeatother
% \IEEEpubid{\parbox{\columnwidth}{
%     Network and Distributed Systems Security (NDSS) Symposium 2021\\
%     21-24 February 2021, San Diego, CA, USA\\
%     ISBN 1-891562-61-4\\
%     https://dx.doi.org/10.14722/ndss.2021.21xxx\\
%     www.ndss-symposium.org
% }
% \hspace{\columnsep}\makebox[\columnwidth]{}}

\IEEEtitleabstractindextext{
\begin{abstract}
Cryptocurrencies are rapidly finding application in areas such as Real Time Gross Settlements and Payments. Ripple is a cryptocurrency that has gained prominence with banks and payment providers. It solves the Byzantine General's Problem with its Ripple Protocol Consensus Algorithm (RPCA), where each server maintains a list of servers, called the Unique Node List (UNL), that represents the network for that server and will not collectively defraud it. The server believes that the network has come to a consensus when servers on the UNL come to a consensus on a transaction. 

In this paper we improve Ripple to achieve better speed, security and last mile connectivity. We implement guidelines for resilience, robustness, improved security, and efficient information propagation (IP). We enhance the system to ensure that each server receives information from across the whole network rather than just from the UNL members. We introduce the paradigm of UNL overlap as a function of IP and the trust a server assigns to its own UNL. Our design makes it possible to identify and mitigate some malicious behaviours including attempts to fraudulently Double Spend or stall the system. We provide experimental evidence of the benefits of our approach over the current Ripple scheme. We observe $\geq 99.67\%$ reduction in opportunities for double spend attacks and censorship, $1.71x$ increase in fault tolerance to $\geq 34.21\%$ malicious nodes, $\geq 4.97x$ and $98.22x$ speedup and success rate for IP respectively, and $\geq 3.16x$ and $51.70x$ speedup and success rate in consensus respectively.
% Structure to be followed while implementing various lists
% By decentralising systems and implementing checks and balances, the systems are protected from single point of failure vulnerabilities such as DDoS attacks, hacking, spoofing and coercion that most centralized systems are vulnerable to. However, this decentralisation introduces other vulnerabilities such as sybil (51\%.) attacks and eclipse attacks which need to be addressed/are being addressed in various ways in the current generation of crypto currencies and distributed blockchain systems. In this paper we work on one such cryptocurrency, Ripple, and help improve its consensus process as a way of solving the Byzantine General's Problem. By helping systematize and set up strategies/guidelines for UNL selection and updation for Ripple Consensus Protocol, utilising and inspired by p2p DHT 
\end{abstract}
}

\maketitle
\IEEEdisplaynontitleabstractindextext
\IEEEpeerreviewmaketitle

\IEEEraisesectionheading{\section{Introduction}}
% \begin{enumerate}
%     \item Blockchain, cryptocurrencies, usage, growth, popularity, etc.
%     \item Unfortunately, there are multiple limitations of general block chains, explain all of them. (only discuss limitations that you fix.. nothing more)
%     \item In this paper we address these problems in Ripple ... boast about ripple.
%     \item high level view about ripple. write quite a bit about why Ripple is a promising solution for all applications in (1) ... need to tell the reader... why are we working with ripple.
%     \item At a very high level, how do you fix these problems in this paper. USP.
%     \item Specify a bit more details about how you are fixing Ripple.... not going into any terminology, no technical stuff, understandable by a Grad student. 
%     \item Summarize all your contributions. (Kelips...)... itemized
%     \item Structure.
% \end{enumerate}

\IEEEPARstart{A} {\bf cryptocurrency} is a medium of exchange implemented as a distributed system. It uses blockchain technology to prevent spending of the same resource more than once (the Double Spending Problem). Transactions are stored in a chain that becomes increasingly resistant to modifications as new transactions are added. Some popular cryptocurrencies are Bitcoin~\cite{nakamoto2008bitcoin}, Ethereum~\cite{wood2014ethereum}, Ripple~\cite{schwartz2014ripple}, Tether~\cite{tetherWhitePaper}, Libra~\cite{librabanostate}, LiteCoin~\cite{liteCoinWhitePaper}, Monero~\cite{cryptonote2013v2}\cite{moneroResearchLab} and IOTA~\cite{popov2017iota}. They are used for Real Time Gross Settlements and Payments. 

\justify {\bf \vspace{-0.75mm} Cryptocurrencies, in general, face challenges} in achieving fast, secure and correct agreement on the validity of transactions (consensus) due to their distributed nature. They attempt to solve these using Proof-of-Work~\cite{nakamoto2008bitcoin}\cite{wood2014ethereum}, Proof-of-Stake~\cite{buterin2017casper}, Proof-of-Elapsed-Time~\cite{olson2018sawtooth}, Practical Byzantine Fault Tolerance~\cite{librabanostate}, and Proof-of-Authority mechanisms. However, some challenges may arise due to malicious or benign network issues like downtime and increased latencies. Additionally malicious nodes could throttle or manipulate information flow, or selectively share wrong information. These open the system to the Byzantine General's Problem~\cite{lamport1982byzantine} {\bf (BGP)} and resulting in attacks like Double Spending causing a blockchain fork, consensus delay, or halt. These negate the benefits of using blockchains.

{\flushleft \vspace{-0.75mm} Cryptocurrencies} need to solve these blockchain challenges {\bf to have widespread adoption and success}. They also need to be fast, reliable, highly secure, and resilient and achieve last mile connectivity. For these we look at one such cryptocurrency, {\bf Ripple}~\cite{rippleWebsite}\cite{schwartz2014ripple} that has gained popularity with banks and payment providers. It helps seamlessly and globally transfer value in seconds, much faster and more energy efficient than Bitcoin. These benefits and efforts to integrate with global financial systems make Ripple promising.

{\flushleft \vspace{-0.75mm} \bf Ripple} solves the BGP using its Ripple Protocol Consensus Algorithm (RPCA)~\cite{schwartz2014ripple}. Each server maintains its own Unique Node List (UNL), a list of servers which it believes will not collectively defraud it and represent the entire network. Servers need to maintain a minimum overlap between UNLs to ensure the correctness of consensus, and prevent network partitions {\bf (NPs)} and blockchain forks. This overlap defines the consensus threshold for votes (proposals).

{\flushleft \vspace{-0.75mm} \bf Limitations of Ripple:} Currently Ripple has some technical challenges. {\bf (1)} There are no guidelines or fool-proof ways to ensure the minimum UNL overlap. 
{\bf (2)} This overlap can be increased and consensus thresholds reduced, achieving consensus sooner without affecting security.
{\bf (3)} The relatively slow rate of information propagation affects consensus' convergence rate and can be exploited. It can be improved and be made more robust, resilient and efficient. 
{\bf (4)} There is no way to ensure that a server receives all the transactions in the network, let alone optimally. Connectivity (including last mile connectivity) and complete network exposure in the face of network issues and malicious behaviour is not assured. Finally,
{\bf (5)} a transaction that is set to undergo consensus (candidate transaction) propagates across the UNL to a node's immediate neighborhood via pull. It can take several consensus rounds for network-wide propagation.

\begin{table}[!t]
\vspace{-0.5mm}
\caption{\small Mapping: 'Limitations of Ripple' to 'Our Contribution'} \label{table:limitationContribution}
\centering{}%
\begin{tabular}{|c|c|c|c|c|c|}
\hline
 & \multicolumn{5}{c|}{Limitation} \tabularnewline
\cline{2-6}
Contribution & 1
& 2
& 3
& 4 & 5\tabularnewline
\hline
\rowcolor{colorLight1Blue}
Modified consensus algorithm & $\times$ & $\surd$ & $\surd$ & $\surd$ & $\surd$ \tabularnewline
\hline 
Guidelines for UNL & $\surd$ & $\surd$ & $\surd$ & $\surd$ & $\times$ \tabularnewline
\hline 
\rowcolor{colorLight1Blue}
UNL overlap & $\surd$ & $\surd$ & $\times$ & $\times$ & $\times$ \tabularnewline
\hline 
\end{tabular}
\vspace{-4mm}
\end{table}

{\flushleft \vspace{-0.75mm} \bf Our Contributions:} The above issues can compromise security, correctness of consensus, affect the rate and quality of information propagation and convergence to a consensus in the face of churn, attacks, network issues and other blockchain challenges. We present solutions to these challenges of Ripple, blockchain, and BGP. We bring about an increase in {\em speed, security and last mile connectivity}. We {\bf (a)} apply a modified consensus algorithm {\bf (MCA)} and address limitations 2, 3, 4 and 5, {\bf (b)} implement guidelines for UNLs and address limitations 1, 2, 3 and 4, {\bf (c)} introduce UNL overlap as a function of information propagation and reputation/ trust value of a node's UNL and address limitations 1 and 2. Table~\ref{table:limitationContribution} maps our contributions to limitations {\bf 1} to {\bf 5}.

{\flushleft \vspace{-0.75mm} \bf Our contributions bring the following impacts.} Our UNL design and MCA propagate information more freely and to the entire network. The UNL design ensures sufficient overlap and thus provable security~\cite{armknecht2015ripple}. It ensures that all nodes receive information from the entire network. We thus prevent attacks due to NPs, blockchain forks, consensus delay or halt, when all the nodes do not receive information or reach consensus. Thereby we tackle Double Spends, Throttling and Censorship, and BGP. Our UNL design also ensures high availability, churn resilience and last mile connectivity. It along with the concept of overlap helps reduce minimum consensus thresholds. Our contributions improve speed and reduce transaction time. They result in a fast, efficient, reliable, highly secure, and resilient system.

{\flushleft \vspace{-0.75mm} \bf Evaluation:} We simulate across multiple systems, consensus algorithms, and consensus, information and maximum shortest path modes. Adverse situations simulated include network issues and malicious nodes spread across the network, and/ or eclipsing a node. We observe that our approaches compared to RippleNet (having random UNLs) impact 
 {\bf (1)} Security:
$\geq 99.67\%$ reduction in double spend and censorship attack vectors.
$1.71x$ increase in fault tolerance to $\geq 34.21\%$ malicious nodes. 
{\bf (2)} Information propagation:
$\geq 4.97x$ speed up, 
$98.22x$ success rate network-wide,
$0.05x$ avg received messages (msgs),
$0.79x$ avg sent msgs,
$419.2x$ rate of having $\leq 3$ hops between source and destination.
{\bf (3)} Consensus: 
$\geq 3.16x$ speed up,
$51.70x$ success rate,
$0.42x$ avg received msgs,
$1.41x$ avg sent msgs.

{\flushleft \vspace{-0.75mm} \bf Organization:} This paper contains 6 sections. First, we briefly describe prior works -- Ripple and Kelips.  Then we list related works in Section \ref{sec:3}. In Section \ref{sec:4}, we  describe  our work -- UNL guidelines and overlay structure, and the modified consensus algorithm. We analyse information propagation and a security model. We introduce and analyse UNL overlap as a function of information propagation and trust value. We cover safety and liveness properties. Section \ref{sec:5} deals with experimental work, set up, analysis and interpretations. We then present the areas of our future work in Section~\ref{sec:futureWork} and conclude the paper in Section~\ref{sec:conclusion}.

\section{Background} \label{sec:background}
\vspace{-0.5mm}
We first describe elements of Ripple~\cite{schwartz2014ripple} and its components.

\vspace{-0.75mm}
\subsection{ Ripple \cite{armknecht2015ripple} \cite{rippleConsensusAdvanced} \cite{rippleTechTalk} \cite{schwartz2014ripple}}

{\vspace{-0.5mm} Each} {\bf server} running the Ripple Server software (implementing the RPCA) participates in consensus. It maintains a {\bf ledger} that records the amount of currency in each user's account and represents the network's "ground truth". It also maintains its own Unique Node List (UNL), a list of servers which it believes will not collectively defraud it and represent the entire network. Each individual UNL member is not required to be a trusted member for this. 

All nodes apply the RPCA  every few seconds to maintain the correctness of and to achieve consensus. Once consensus is reached, the current ledger is updated with transactions passing consensus and is considered closed. It becomes the {\bf last-closed ledger} and represents the current network state. We now describe the RPCA. A similar process is followed for validation post consensus.

{\flushleft \vspace{-1mm} \bf Ripple Consensus Algorithm (RPCA) \cite{rippleConsensusAdvanced}\cite{schwartz2014ripple}:}
The RPCA proceeds in rounds. In each round:
\begin{enumerate}
\item Each server takes valid transactions seen prior to the start of consensus, that it has not applied to the ledger. This includes new transactions initiated by the server's users and those pending from a previous consensus round. It makes them public as a list ({\bf candidate set}).
\item Each server amalgamates the candidate sets of servers in its UNL, and votes on the veracity of transactions by sending proposals (yes votes). 
\item The server only considers proposals from servers in its UNL. Transactions that receive more than a minimum percentage of yes votes are passed on to the next round while those that do not are either discarded, or included in the candidate set for the next consensus round.
\item The final consensus round requires a minimum $80\%$ of the UNL agreeing on a transaction. Consensus is reached on transactions passing this requirement.
\end{enumerate}

{\flushleft \vspace{-0.75mm} \bf Agreement and forking:}
Agreement ensures that all nodes agree to the same ledger version. This ensures absence of "forks" where more than one different version of the ledger exists for certain subsets of nodes. Presence of fork(s) results in the {\bf double spending problem}. After a formal analysis, the relation $w_{u,v} \geq 2(1-\rho)$~\cite{armknecht2015ripple} was identified {\em to ensure absence of forks}~\cite{armknecht2015ripple}. Here $w_{u,v}$ is the minimum UNL overlap between any two nodes and $\rho$ is the consensus threshold.

{\flushleft \vspace{-0.75mm} \bf Transaction ordering} for transactions is done using Account Sequence~\cite{ripple:accountSequence} for each account. This is a property of each account and incremented for each transaction. It is represented by the Sequence field in the transaction's fields~\cite{ripple:transactionCommonFields}.

In the paper we set up UNL guidelines. For this we use Kelips'~\cite{gupta2003kelips} overlay structure and efficient query routing.

\vspace{-0.75mm}
\subsection{Kelips \cite{cloudComputingConcepts} \cite{gupta2003kelips}}
{\vspace{-0.5mm} Kelips} is a peer-to-peer (p2p) Distributed Hash Table (DHT) with O(1) file look-up complexity. Consider a distributed system of $N$ nodes. The nodes are divided into $\sqrt{N}$ buckets (affinity groups) of size $\sqrt{N}$ each. Each node is hashed to a unique affinity group. It stores socket address of nodes in its affinity group in a list and those of $w$ nodes each for other affinity groups in another. This provides alternate routes for information to propagate in the face of faults and failures.

{\flushleft \vspace{-0.75mm} \bf Why Kelips:} Popular DHTs such as Kelips~\cite{gupta2003kelips}, Kademlia~\cite{maymounkov2002kademlia}, Content Addressable Network (CAN)~\cite{ratnasamy2001can}, Pastry~\cite{rowstron2001pastry}, Chord~\cite{stoica2001chord},  Tapestry~\cite{zhao2004tapestry} have lookup complexities of $O(1)$, $O(log(n))$, $O(n^{1/d})$, $O(log(n))$,  $O(log_{2}(n))$, $O((b-1)\times log_b(n))$ respectively. Kelips based information propagation for Ripple would propagate information and achieve consensus sooner. It would be churn resilient and tolerate faults and failures. Kelips' slightly higher memory overhead is negligible per today's standards. Thus Kelips.

\section{Related Works} \label{sec:3}
{\flushleft \vspace{-0.75mm} \bf Bitcoin and DNS seeds:} Bitcoin clients rely on Domain Name Server (DNS) seeds instead of DHTs for the initial peer selection (fresh or after long disconnection). The seeds resolve to a list of IPs of running nodes. If this fails, they fallback to a hard-coded list of IPs pointing to stable nodes. The previous list is used for subsequent re-connections.

{\flushleft \vspace{-0.75mm} \bf Ethereum and Kademlia:} Ethereum~\cite{wood2014ethereum} and its spawned cryptocurrencies use Kademlia~\cite{maymounkov2002kademlia} for peer selection and message passing. We notice some challenges. {\bf (1)} The inefficiency due to Kademlia's $O(log(n))$ look up complexity cascades to message passing and voting for future  PoS~\cite{buterin2017casper} implementation. It is worsened by several round trips. {\bf (2)} The distance between peers is virtual and not per actual network topology. Thus systems and messages are not connected efficiently. Thus performance is affected due to inefficiencies in message passing, dissemination and voting, and with scale (with at least $O(log(n))$ complexities). {\bf (3)} Bootstrapping the peer table with one or a few peers is inefficient. It can be circumvented by starting with sufficient peers. Protocols to maintain the peer table too should be given adequate attention. These areas can be vulnerable to table poisoning, overflow attacks, and eclipse attacks~\cite{ethereumAttacksSurvey}\cite{eclipseEthereumFalseFriendsAttack}\cite{lowresourceEclipseEthereumAttack}. {\em Our proposal is different from the existing Kademlia solution.} Since Ethereum uses PoW and later PoS, Kademlia is used strictly for peer selection and message passing. We use Kelips, a more appropriate DHT, not just for peer selection and message passing but also for systematizing the UNL for consensus voting and improving UNL overlap. We provide improvements in speed, security and last mile connectivity.

{\flushleft \vspace{-0.75mm} \bf BitTorrent and Mainline DHT:} BitTorrent clients~\cite{BitTorrentWebsite}\cite{uTorrentWebsite}\cite{TransmissionWebsite}\cite{rTorrentWebsite} use Mainline DHT~\cite{wang2013mainlinedht}, a Kademlia-based~\cite{maymounkov2002kademlia} DHT, for peer selection.

{\flushleft \vspace{-0.75mm} \bf IOTA and auto-peering:} IOTA does not have automated peer selection. Post Coordicide~\cite{popov2020coordicide}, it is working to incorporate auto-peering and peer discovery~\cite{popov2020coordicide} using DNS seeding like process for bootstrap and ping-pong for liveness.

\section{Our Work: SISSLE In Consensus - Ripple} \label{sec:4}
We highlight properties of distributed systems and blockchains below. We also highlight our impact.

{\flushleft \vspace{-1mm} From} CAP theorem~\cite{brewer2000CAP}, blockchains can achieve two of the three: scale, security, and decentralization~\cite{medium:scalabilityTrilemma}. We improve scale and security at negligible cost to decentralisation.

Distributed systems trade-off between safety and liveness. Safety guarantees “a bad thing never happens”. Liveness guarantees “something good eventually happens”. Ripple demonstrates liveness using its simulator~\cite{simCode}. Safety (no forks or double spends {\bf [DSs]}) has been analysed previously~\cite{armknecht2015ripple}. Ripple's implementation has no fool-proof mechanism to ensure safety. We improve it's safety and liveness.

\subsection{\vspace{-0.75mm} A Brief Summary of Our Contributions}
We present our work improving speed, security and last mile connectivity for Ripple and other blockchains.

{\flushleft \vspace{-0.75mm} \bf (1)} We introduce peer-to-peer (p2p) inspired network overlays, present a modified version of the UNL, and introduce a new list, the Trustee Node List (TNL).

{\flushleft \vspace{-0.75mm} \bf (2)} We present our modified consensus algorithm {\bf (MCA)} where information propagates more freely. Propagation is to the UNL and TNL via push (all purposes). The existing algorithm's information propagation is more restricted and only to the UNL. It is via pull for the generation of transaction sets to be voted on and via push for consensus voting.

{\flushleft \vspace{-0.75mm} \bf (3)} We show the soundness and benefits of our approach. We formally analyse and list possible {\bf information propagation (IP)} paths and create a security model. We list the possible malicious behaviours and some solutions to them. Using proofs and experiments we show that information propagates in 3 Hops across the entire network as long as - a) each node has at least one genuine and non-faulty node in its UNL or TNL, b) our approaches (Kelips-like overlay and MCA) are implemented, and, c) in the presence of $<(c+1)\times(\sqrt{N}-1)$ faults. We show how our proposed approaches improve last mile connectivity, and mitigate network partitions (NPs), DS attacks and censorship. We connect IP, UNL overlap and consensus thresholds.

{\flushleft \vspace{-0.75mm} \bf (4)} We analyse our approach's impact on safety and liveness.

These contributions solve Ripple's limitations in an overlapping fashion as described earlier in Table~\ref{table:limitationContribution}

\vspace{-2.75mm}
\subsection{P2P inspired overlay for trust and consensus} \label{subsec:5A}
{\flushleft \vspace{-0.75mm} \bf Guidelines and overlay structure for UNL:} Consider a Kelips-like~\cite{gupta2003kelips} network overlay with $N$ nodes running the MCA (described immediately after). These are divided into $\sqrt{N}$ affinity groups, each of size $\sqrt{N}$.

We propose guidelines for trust lists, UNL and TNL, maintained by individual servers. Trust lists contain all nodes connected to that server. A server's UNL contains two parts: \textbf{UNL-A} or a list of socket addresses of nodes belonging to its affinity group. Its size is $\sqrt{N}-1$. \textbf{UNL-B} or a list of socket addresses of $c$ nodes from each of the $\sqrt{N}-1$ affinity groups other than its own (foreign affinity groups). Its size is $c*(\sqrt{N}-1)$ where ideally $c \geq 3$. 

A server's Trustee Node List (TNL) contains all servers which have it in their UNL. They depend on its messages towards candidate set generation and consensus. Our guidelines are implementable manually or as a dynamic UNL.

\vspace{-2.75mm}
\subsection{Modified Consensus Algorithm (MCA):}
{\vspace{-0.75mm} We} propose slight modifications to the RPCA ~\cite{schwartz2014ripple} in IP and consensus thresholds. For all other purposes, the algorithm is the same.
Each server takes the following steps during each consensus round $i$:

{\flushleft \vspace{-0.5mm} \bf Stage 1: Candidate Set Generation}

{\flushleft \vspace{-0.75mm} \bf Sub-round A:}  The server adds valid transactions to its candidate set. These include transactions pending from previous consensus rounds and fresh ones from its users.

{\flushleft \vspace{-0.75mm} \bf Sub-round B:} The server declares its candidate set to servers in its UNL (and optionally TNL) once.

{\flushleft \vspace{-0.75mm} \bf Sub-round C:} 
{\em Outbound:} The server forwards candidate sets received from other servers as a set of sets every $x$ seconds. Each candidate set is forwarded once.

{\flushleft \vspace{-0.75mm} \em Inbound:} The server vets and assimilates transactions in candidate sets of servers in its UNL (and optionally TNL). It may also consider those of other servers.

{\flushleft \vspace{-0.75mm} \em Duration:} Sub-round C runs for $y$ seconds optimised for visibility of most or all of the network's transactions. 
% {\color{red} $y$ includes time for (1) min 3 hops\footnote{In this paper, we imply 'hops' to be vertex hops.} of IP in our overlay and (2) distributed nature of consensus (factor of 2). Thus $y \geq 6 \times x$.} 
If some servers are unable to have network-wide exposure, their inputs are considered in the next consensus round. If this is a large majority, consensus wouldn't pass. Then $y$ needs tweaking. Our multi-hop\footnote{In this paper, we imply 'hops' to be vertex hops.} overlay relays information indirectly providing this exposure.

{\flushleft \vspace{-0.75mm} \em Note A:} Transactions involve a processing fee with greater priority to higher fee. Transactions are propagated in the descending order of this fee for constrained bandwidths.

{\flushleft  \vspace{-0.75mm} \em Note B:} Attempts to flood the network with transactions would erode the attacker's wealth if legitimate transactions have a lower processing fee. This is the best penalty. It might inhibit processing for legitimate users and force them to pay greater fee. AI agents can be used to prevent such attacks.

{\flushleft \vspace{-0.75mm} \em Note C:} Network-wide IP takes multiple consensus rounds currently and one round in our approach. (1). Our approach's mandatory wait time ($y$) is lesser than the sum of current wait times across multiple consensus rounds. (2). Entire consensus rounds (candidate set generation and voting) spent waiting for network-wide transaction receipt are saved. Consensus is sped up.

{\flushleft  \vspace{-0.5mm} \bf Stage 2: Consensus Voting}

{\flushleft \vspace{-1mm} This stage has several sub-rounds. In each sub-round:}
\begin{itemize}
    \item The server sends its proposals (yes votes) to servers in its TNL (and optionally UNL).
    \item The server forwards other servers' proposals to servers in its UNL and TNL once per proposal and in sets.
    \item The absolute threshold for mathematical certainty decreases after each sub-round with greater IP and thus network overlap.
    \item The threshold for inclusion of transactions to the next sub-round increases with each sub-round, as in the RPCA~\cite{schwartz2014ripple}. This threshold need not be more than the absolute threshold for mathematical certainty.
    \item Only proposals (votes) from the UNL are considered.
\end{itemize}

\vspace{-0.75mm}
\subsection{Formal Analysis: Information Propagation (IP), Security Model, UNL Overlap and Consensus Threshold} \label{subsec:formalAnalysis}

We analyse IP and security in our approach. We connect UNL overlap to IP and the percentage of non-byzantine nodes in the UNL. It extends to consensus thresholds for provable security and achieving mathematical certainty of the absence of forks~\cite{armknecht2015ripple}.

\begin{table}[!t]
% \vspace{-3mm}
\caption{A classification of propagation paths paths with source and destination in same and separate affinity groups.} \label{table:pathsSameAffinityGroup}
\centering{}%
\begin{tabular}{|c|c|c|c|}
\hline 
Path
& \begin{tabular}{@{}c@{}}No. of\\Hops\end{tabular}
& \begin{tabular}{@{}c@{}}No. of\\Paths\end{tabular}  
& P(X) \tabularnewline
\hline
\rowcolor{colorLight1Blue}
\multicolumn{4}{|c|}{\bf Same affinity group} \tabularnewline
\hline 
$A_s \rightarrow A_d$ & $1$ & $1$  & $1$ \tabularnewline
\hline 
\rowcolor{colorLight1Blue}
$A_s \rightarrow A_i\rightarrow A_d$ & $2$ & $\floor*{\sqrt{N}-2}$  & $1$ \tabularnewline
\hline 
$A_s \rightarrow B_i \rightarrow A_d$  & $2$ & - & $<1$ \tabularnewline
\hline 
\rowcolor{colorLight1Blue}
$A_s \rightarrow B_i \rightarrow B_j \rightarrow A_d$  & $3$ & $\geq c^2\times(\floor*{\sqrt{N}-1})$  & $1$ \tabularnewline
\hline 
$A_s \rightarrow B_i \rightarrow C_j \rightarrow A_d$  & $3$ & -  & $<1$ \tabularnewline
\hline 
\rowcolor{colorLight1Blue}
$A_s \rightarrow B_i \rightarrow A_j \rightarrow A_d$  & $3$ & $\geq c\times(c-1)$  & $1$ \tabularnewline
\hline 
\begin{tabular}{@{}c@{}}$A_s \rightarrow A_1\rightarrow \ldots$ \\ $\rightarrow A_l \rightarrow A_d$\end{tabular}  & $l$ & $\Myperm[\floor*{\sqrt{N}-2}]{l}$  & $1$ \tabularnewline
\hline
\rowcolor{colorLight1Blue}
\multicolumn{4}{|c|}{\bf Separate affinity groups} \tabularnewline
\hline 
$A_s \rightarrow B_d$ & $1$ & $1$  & $1$ \tabularnewline
\hline
\rowcolor{colorLight1Blue}
$A_s \rightarrow A_j \rightarrow B_d$ & $2$ & $\geq c$  & $1$ \tabularnewline
\hline
$A_s \rightarrow B_i \rightarrow B_d$ & $2$ & $\geq c$  & $1$ \tabularnewline
\hline
\rowcolor{colorLight1Blue}
$A_s \rightarrow C_i \rightarrow B_d$ & $2$ & -  & $<1$ \tabularnewline
\hline
$A_s \rightarrow B_1 \rightarrow B_2 \rightarrow B_d$ & $3$ & $\geq c\times (\floor*{\sqrt{N}-2})$  & $1$ \tabularnewline
\hline
\rowcolor{colorLight1Blue}
$A_s \rightarrow A_i \rightarrow B_j \rightarrow B_d$ & $3$ & $\geq c \times (\floor*{\sqrt{N}-1})$  & $1$ \tabularnewline
\hline
$A_s \rightarrow C_i \rightarrow C_j \rightarrow B_d$ & $3$ & $\geq c^2 \times (\floor*{\sqrt{N}-2})$  & $1$ \tabularnewline
\hline
\rowcolor{colorLight1Blue}
$A_s \rightarrow C_i \rightarrow B_k \rightarrow B_d$ & $3$ & $\geq c^2 \times (\floor*{\sqrt{N}-2})$  & $1$ \tabularnewline
\hline
$A_s \rightarrow C_i \rightarrow A_j \rightarrow B_d$ & $3$ & -  & $<1$ \tabularnewline
\hline
\rowcolor{colorLight1Blue}
\begin{tabular}{@{}c@{}}$A_s \rightarrow C_i \rightarrow C_j$ \\ $\rightarrow A_k\rightarrow B_d$\end{tabular}  
 & $4$ & $\geq c^3 \times (\floor*{\sqrt{N}-2})$  & $1$ \tabularnewline
\hline
\begin{tabular}{@{}c@{}}$A_s \rightarrow C_i \rightarrow C_j$ \\$\rightarrow B_k\rightarrow B_d$\end{tabular}   & $4$ & \begin{tabular}{@{}c@{}}$\geq c^2 \times (\floor*{\sqrt{N}-1})$ \\ $\times (\floor*{\sqrt{N-2}})$\end{tabular} & $1$ \tabularnewline
\hline
\rowcolor{colorLight1Blue}
\begin{tabular}{@{}c@{}}$A_s \rightarrow B_1 \rightarrow \ldots$ \\ $\rightarrow B_l \rightarrow B_d$ \end{tabular}  
& $l+1$ & $\geq c\times \Myperm[\floor*{\sqrt{N}-2}]{l}$  & $1$ \tabularnewline
\hline
\end{tabular}
\vspace{-6mm}
\end{table}

{\flushleft  \vspace{-0.75mm} \bf A.) \ul{IP:}} For IP we implement our approaches in a Kelips-like overlay. In our 3 Hop Claim we claim that there is a path of length 3 hops connecting any two nodes (maximum $2$ degrees of separation\footnote{We define degree of separation as the number of nodes between source and destination nodes in a path.}) in this overlay having a minimum of $< (c+1) \times \sqrt{N}-1$ faults. Towards this claim we list paths of length $\leq 3$ hops, the total number of such paths and their probability. We then build a security model for Ripple including vulnerabilities and their solutions. Next we prove this claim formally and experimentally. We then present benefits in last mile connectivity and mitigation of NPs, DSs, blockchain forks and censorship.

{\flushleft  \vspace{-0.75mm} \em Formal listing of paths for IP:} We present the same in  Table~\ref{table:pathsSameAffinityGroup}. A generic representation of these paths is provided for two cases: when the source and destination node's affinity group is the same and when it is different. Notations: $A$, $B$ and $C$ are different affinity groups. Nodes in an affinity group have unique subscripts. Subscripts: $s$ and $d$ are for source and destination nodes; $i$, $j$ and $k$ are for different and random nodes; $1$ to $l$ are individual nodes. P(X) is the probability of the particular path's existence.

Factors in number of paths estimation: $\bullet$ $c$ is when a node accesses a node in its UNL or TNL but not in its affinity group. $\bullet$ $-$ is for paths having probability less than $1$. Their estimation is complex. $\bullet$ $\sqrt{N}$ is when a node accesses one of the $\sqrt{N}-1$ nodes in its own affinity group or one of the other $\sqrt{N}-1$ affinity groups. $\bullet$ $\floor*{}$ is floor function. $\bullet$ $\geq$ as $c$ accounts only for nodes in the UNL and not the TNL.

Number of paths of length $\leq 3$ between source (S) and destination (D) nodes: $\geq (c^2 \times \sqrt{N}-1) + (c\times(c-1))$ paths when S and D in same affinity group. $\geq 2\times c \times (c+1) \times \sqrt{N}-2$ paths when S and D in different affinity groups.   

{\flushleft \vspace{-0.5mm} \bf B.) \ul{Security Model:}} Several vulnerabilities in cryptocurrencies have been documented~\cite{armknecht2015ripple}\cite{informalSurvey:apriorit}\cite{ethereumAttacksSurvey}\cite{chen2020survey}\cite{conti2018survey}\cite{hasanova2019survey}\cite{lin2017survey}\cite{sayeed2018effectiveness}.  Given this paper's domain, we focus on vulnerabilities at the network and protocol level and not account or client level.  Some of these apply to Ripple: {\bf(1) Direct:} DS Attack~\cite{karame2012two}\cite{rosenfeld2014analysis}, race attack~\cite{bitcoinit:irreversibleTransactions}, 51$\%$ attack (Sybil~\cite{douceur2002sybil} or Bribery~\cite{bonneau2016buy}), DDoS~\cite{johnson2014game}\cite{vasek2014empirical}, transaction flooding, Eclipse attack~\cite{heilman2015eclipse}\cite{eclipseEthereumFalseFriendsAttack}\cite{lowresourceEclipseEthereumAttack}, network partitioning, routing attack~\cite{apostolaki2017hijacking}\cite{igiglobal:routingAttack}, censorship~\cite{censorshipDetectors}\cite{ethblog:problemCensorship}, {\bf(2) Indirect:} Proof-of-Work (PoW) based attacks (Finney~\cite{bitcointalk:finneyAttack}\cite{stackexchange:finneyAttack}, Vector76~\cite{bitcointalk:vector76}\cite{reddit:vector76Attack} and Brute Force~\cite{conti2018survey}), Proof-of-Stake (PoS) based attack (Nothing-At-Stake~\cite{ethwiki:PoSFAQ}), low voter turnout exploit~\cite{vitalikca:NotesOnBlockchainGovernance}, $51\%$ attack (Goldfinger~\cite{extance2015future}\cite{ethwiki:PoSFAQ}\cite{ren2014proof}), attacks at scale~\cite{steemit:DPOSConsensusAlgorithmMissingWhitePaper}\cite{steemit:ResponseToCosmosWhitePaperOnDPOSSecurity}, block witholding~\cite{bag2016bitcoin}\cite{rosenfeld2011analysisPooled}\cite{tosh2017security}, fork after witholding (FAW)~\cite{kwon2017selfish}, time jacking~\cite{conti2018survey}, {\bf (3) Less likely:} transaction malleability~\cite{andrychowicz2015malleability}\cite{stock2009walowdac}, tampering~\cite{conti2018survey}, {\bf (4) Inapplicable:} selfish mining~\cite{atzei2017survey}, time base and long range attacks~\cite{atzei2017survey}. We are not aware of peer reviewed works on Ripple's security beyond~\cite{armknecht2015ripple}.
% {\color{red} \em Regardless of our model's exhaustiveness in Ripple's context (we believe it is), the vulnerabilities and their solutions we identify enhance security as described in this section.} 

{\flushleft \vspace{-0.75mm} {\bf B.1.)} \em Vulnerability overview in a Ripple like setting (from above and more):} A) Protocol level: 1) nodes tampering messages, or 2) nodes sending erroneous or malformed transactions, proposals and last closed ledgers (LCLs). B) Network level: 1) nodes throttling or dropping information and then leveraging its effects, 2) DDoS and transaction flooding by nodes or 3) Border Gateway Protocol (BGP) hijack. C) Other: 1) 51$\%$ attack (Sybil or Bribery), 2) inadequate UNL overlap or poor UNL configuration, 3) attacks at scale~\cite{steemit:DPOSConsensusAlgorithmMissingWhitePaper}\cite{steemit:ResponseToCosmosWhitePaperOnDPOSSecurity} 4) real world transaction execution on partial digital confirmation, and subsequent digital transaction failure. 5) Transaction flooding by accounts is at protocol and account level.

{\flushleft  \vspace{-0.75mm} {\bf B.2.)} \em We now rule out some attack vectors.} 1) Public key cryptography prevents {\em malicious nodes from tampering messages or generating transactions} (transaction generation by accounts). This prevents message tampering and transaction flooding.
2) Priority is given to transactions with higher processing fee. Thus {\em transaction flooding} can only raise this fee. This attack is economically infeasible as such transactions erode the attacker's wealth if processed. This attack is at an account level and out of the scope of this security model. 
3) {\em Partial confirmation based attacks} can be prevented by waiting for the Last Validated Ledger (LVL).
Finally, 4) messages sent by nodes are cryptographically signed. It is possible to trace {\em malformed messages and erroneous LCLs} to nodes. The community (currently Ripple Labs, Inc.) can identify and analyse such nodes' behaviour and choose not to trust them~\cite{rippleTechTalk}. Thus nodes behaving so do it at their own detriment. We call this solution by Ripple as {\bf Solution A}. Information to identify such behaviour may not be visible network-wide. Our approaches ({\bf Solution B}) solve this. 

{\flushleft \vspace{-0.75mm} {\bf B.3.)} \em Some vulnerabilities in detail:} Vulnerabilities are denoted (V$x$) and their solutions (S$x$) where $x$ is a number

{\flushleft \vspace{-1mm} \bf (V1.)} Poor UNL design causes poor overlap, DSs, forks~\cite{armknecht2015ripple}.

{\flushleft  \vspace{-1mm} \bf (V2.)} is a node not declaring some proposals in a consensus sub-round but declaring them in the next sub-round. It is an anomaly\footnote{Evident from point 3 of RPCA}. It can be malicious or benign\footnote{Throughout this paper benign implies non-malicious and genuine} due to network issues like temporary downtime and increased connection latencies. Currently, there is no way to differentiate between these at a node level or to mitigate the {\em benign challenges}. Due to Byzantine possibilities Solution A may face difficulty in identifying malicious nodes exhibiting V2.

{\flushleft  \vspace{-1mm} \bf (V3.)} is declaration of an erroneous LCL\footnote{The LCL does not have the same weight as the LVL. This vulnerability's impact is relatively negligible.}. It can be malicious or benign due to network issues like temporary downtime and increased connection latencies and with the distributed nature of consensus affecting candidate set generation and proposals. Currently, there is no way to mitigate the benign aspects of V3. Solution A may find it difficult to distinguish between nodes exhibiting V3 maliciously or benignly.    

{\flushleft \vspace{-1mm} Nodes can maliciously throttle or drop messages.} {\bf (V4.)} is malicious nodes throttling or dropping messages between two parts of the network. {\bf (V5.)} is malicious nodes sending different messages to two different parts of the network and throttling or dropping them across the two parts. The worst case in these is malicious dropping of all the messages. V4 and V5 can {\em censor}, cause {\em NPs} and {\em DS Attacks}. Eclipse attacks and slower IP worsen it.

{\flushleft \vspace{-1mm} \bf (V6.)} is a node declaring proposals for transactions and not including them in its LCL. It can be benign (due to consensus failure on the transaction) or malicious. {\em Currently, there is no way to differentiate between these.}

{\flushleft \vspace{-0.75mm} {\bf B.4.)} \em Out of scope of this model:} 1) 51$\%$ attack. 2) Attacks at scale~\cite{steemit:DPOSConsensusAlgorithmMissingWhitePaper}\cite{steemit:ResponseToCosmosWhitePaperOnDPOSSecurity}. 3) Human error in UNL configuration (not poor design). 4) BGP hijack 5) Vulnerability 6 (V6).

{\flushleft \vspace{-0.75mm} {\bf B.5.)} \bf Adversary model:} V1 is a design vulnerability. V2-V5 are possible malicious behaviours a node can demonstrate at the network and protocol layer. Consider adversaries exhibiting behaviour V2-V5 and freely choosing their UNL.

{\flushleft \vspace{-0.75mm} {\bf B.6.)} \bf Security Claims:} 

{\flushleft \vspace{-1mm} \bf \ul{(S1.)}} Our UNL guidelines mitigate V1.

{\flushleft \vspace{-1mm} \bf \ul{(Base arguments, S4 and S5.)}} Identification and mitigation of V4 and V5 happens if messages reach the entire network. Solution B propagates messages to the entire network if the node is able to send the information to atleast one non-malicious node as proven in the 3 Hop Claim. This mitigates V4 and V5. Messages reach faster and to $100\%$ of the nodes in far more severe cases compared to the existing system. Thus nodes have visibility of each node's messages and sooner. Last mile connectivity is ensured.  NPs, blockchain forks, DS attacks and censorship are mitigated. 

{\flushleft  \vspace{-1mm} \bf \ul{(S2, S3, and Solution A.)}} From base arguments, Solution B mitigates benign aspects of V2 and V3. It becomes easier to identify malicious aspects of V2 and V3 via Solution A. In the absence of Solution B, this identification would have taken longer, may still be uncertain and difficult to prove. Nodes could have claimed benign network and consensus challenges to escape being labelled malicious. 

Thus, {\em Solution B aids in the mitigation of V2, V3, V4 and V5 and increases the efficacy of Solution A.}

{\flushleft \vspace{-0.75mm} \bf B.7.) Security Proofs: \ul{(S1.)}} The cumulative strength of our approaches including UNL guidelines to counter V1 shall be demonstrated across all the sections upto Section~\ref{sec:conclusion}. This includes $100\%$ IP to all the nodes in far more severe circumstances and sooner, last mile connectivity, IP speed up, mitigation of NPs, blockchain forks, DS attacks and censorship, mitigation of V2-V5, increased and provable UNL overlap, reduced consensus thresholds, consensus speed up, improved safety and liveness, increased tolerance to faults and malicious nodes.

{\flushleft \vspace{-1mm} \bf \ul{(Base arguments, S4 and S5.)}} We solve vulnerabilities at the network layer. The $3$ Hop Claim, its corollaries and arguments solve V4 and V5. For ease of explanation consider one half of the network as one node.

\begin{claim} [$3$ Hop Claim] \label{claim:3hopClaim} \vspace{-0.5mm} A node $X$ under attack by nodes trying to throttle or drop messages to and fro the rest of the network can receive and transmit messages across the whole network within $3$ hops as long as it has at least one genuine and non-faulty node $Y$ in its UNL or TNL, our approach (Kelips-like overlay and MCA) is used, and in the presence of $<(c+1)\times(\sqrt{N}-1)$ faults.
\end{claim}
\begin{proof} \vspace{-1.5mm} Consider a system implementing our approaches (Kelips-like overlay and MCA). Due to the MCA, IP is freer and faster. Note, that a message path is bi-directional. 

For proving the claim, it would suffice to prove $\exists$ a path $P$ containing genuine nodes between a node $N$ of the network and $X$ via $Y$ (where $Y$ is the only genuine and non-faulty node connected to $X$) and its length $\leq 3$.

Malicious nodes are throttling or dropping messages. We assume that there is no path via them. We eliminate them. All nodes now discussed are genuine and non-faulty.

{\flushleft \vspace{-1mm} Case 1:} If $N$ is in the UNL of $X$, then $N$ is connected to $X$. This means $N$ is genuine and connected to $X$ $\Rightarrow$ $N$ is $Y$. $\Rightarrow \exists$ a path $N\rightarrow X$ of length $1$.

{\flushleft \vspace{-1mm} Case 2:} If $N$ is not in the UNL of $X$, then $N$ is not in the affinity group of $X$. Then two possibilities exist: (a): $N$ is in the UNL of $Y$ $\Rightarrow \exists$ a path $N\rightarrow Y\rightarrow X$ of length $2$. (b): $N$ is not in the UNL of $Y$. Given the Kelips-like overlay $\exists$ a node $N_i$ in affinity group of $N$ and in the UNL of $Y$. $\Rightarrow \exists$ a path $N\rightarrow N_i \rightarrow Y \rightarrow X$ of length $3$.

The presence of $<(c+1)\times(\sqrt{N}-1)$ faults ensures that no node is fully eclipsed by malicious nodes and that there is one genuine and non-faulty node in each UNL.
\end{proof}

\begin{table}[!t]
% \vspace{-3mm}
\caption{\label{3HopStats} Statistics associated with $3$ hop claim}
\centering{}%
\begin{tabular}{|c|c|c|c|c|c|c|}
\hline 
$N=256$
& \multicolumn{3}{c|} {\begin{tabular}{@{}c@{}}Malicious Nodes \\Eclipsing the\\ Node\end{tabular}}
& \multicolumn{3}{c|}{\begin{tabular}{@{}c@{}}Malicious Nodes \\Randomly Throughout\\ The Network\end{tabular}}
\tabularnewline
\hline 
$c=2$ 
& \begin{tabular}{@{}c@{}}Avg\\dist\end{tabular}
& \begin{tabular}{@{}c@{}}Avg\\$\%$ mal\end{tabular}
& \begin{tabular}{@{}c@{}}Max\\dist\end{tabular}
& \begin{tabular}{@{}c@{}}Avg\\dist\end{tabular}
& \begin{tabular}{@{}c@{}}Avg\\$\%$ mal\end{tabular}
& \begin{tabular}{@{}c@{}}Max\\dist\end{tabular}
\tabularnewline
\hline 
\rowcolor{colorLight1Blue}
SimC & $3.24$ & $7.36$  & $4$ & $3$ & $17.19$ & $3$ \tabularnewline
\hline
SimRM & $2$ & $17.19$  & $3$ & $2$ & $17.19$ & $2$ \tabularnewline
\hline
\rowcolor{colorLight1Blue}
SimK & $2$ & $17.19$  & $3$ & $2$ & $17.19$ & $2$ \tabularnewline
\hline
\end{tabular}
\vspace{-6mm}
\end{table}

{\bf \vspace{-1.5mm} We reinforce this claim experimentally (ie theory with application).} Experiments are run on a $256$ node network over $10000$ seeded cases. We measure the maximum of shortest distances {\bf (MSD)} between nodes and the source of messages (transaction/ proposal). We do so over Ripple's (SimC) overlay, a scaled up version of Ripple (SimRM) with connections $10\%$ greater than a Kelips-like overlay, and our approach (SimK). We present it in table \ref{3HopStats}. The simulation is across two patterns of cases: malicious nodes eclipsing a node, and malicious nodes distributed randomly and evenly in the network. The no. of malicious nodes $<(c+1)\times(\sqrt{N}-1)$. In simulations, malicious nodes drop messages. If successful, all genuine nodes receive messages. Experimentally we observe MSD between nodes and the source reaches a maximum of $3$ for SimK and $4$ for SimC. Messages reach all nodes when the avg no of malicious nodes $<(c+1)\times(\sqrt{N}-1)$. Observation: MSD $>3$ for only SimC in $1175/5000$ cases and when a node is eclipsed by all but one malicious node.

MSD $\leq 3$ with $100\%$ success\footnote{this metric including distance criteria is also known as $Success_2$} for SimK, if 
$\bullet$ total percentage of malicious nodes in the network $\leq 80$ (for $c=2$) and malicious nodes distributed randomly, or $\bullet$ total percentage of malicious nodes $\leq 28$ (for $c=2$) and any particular node eclipsed with all but one connection being malicious.

As a consequence of Claim~\ref{claim:3hopClaim} it is safe to conclude:

\begin{corollary} \label{corollary:100infoprop} \vspace{-0.5mm} It is possible to achieve $100\%$ IP within a max 3 hops so long as each non-malicious node is connected to atleast one non-malicious node, our approaches are implemented and faults $<(c+1)\times(\sqrt{N}-1)$.
\end{corollary}

{\vspace{-1mm} In practice}  $100\%$ IP is achieved in 3 hops in far more severe cases as demonstrated by above experiment. It will be reinforced later in Simulation section. {\em Claim~\ref{claim:3hopClaim} and Corollary~\ref{corollary:100infoprop} mitigate V4 and V5.} They lead to the following benefits:

{\flushleft \vspace{-1mm} \bf Last mile connectivity:} Building upon Claim~\ref{claim:3hopClaim}, we can say: if a server with poor connectivity and at a last mile location can connect to one other non-faulty server, its message gets relayed to the entire network in $3$ hops. The system will reach consensus on that transaction. The server can then receive the LVL. This server acts as a relay and receives the LVL. We ensure last mile connectivity and increase the graph's connectedness. This was not guaranteed before.

{\flushleft \vspace{-1mm} \bf IP speed up:} Our approach provides $100\%$ IP and IP speed up. There are multiple IP routes. Several routes with the lowest latencies help messages reach faster. This is non-redundant and has a limited effect (positive or negative) on message complexity compared to the existing system. 

We now build Claim~\ref{claim:DoubleSpendAttack}. We leverage it and previous claims to solve some key blockchain challenges.

\begin{claim} \label{claim:DoubleSpendAttack} \vspace{-0.5mm} DS Attacks arise due to information asymmetry and can be countered with information availability
\end{claim}
\begin{proof} \vspace{-1.5mm} Consider an account with $100$ credits. Consider two payment transactions: $T1$ of $90$ credits to Alice and $T2$ of $90$ credits to Bob. A malicious actor Mallory can DS so long as transactions $T1$ and $T2$ can be sent to two different NPs $A$ and $B$ respectively and segregation of transactions is maintained across partitions. 

This can be countered if both the partitions $A$ and $B$ receive both the transactions $T1$ and $T2$.

If transaction ordering can be established, then only one transaction passes. If both have the same sequencing, an anomaly can be detected and handled (say, reject both).
\end{proof}

{\vspace{-1mm} DSs} can happen when the Account Sequence~\cite{ripple:accountSequence} is the same (Client-level vulnerability) and with insufficient overlap. At best consensus will fail and the attack goes undetected. At worst DS happens. Our approach makes it possible to detect and handle this during candidate set generation.

{\flushleft \vspace{-1mm} \bf Mitigating NP, blockchain fork, DS attack and censorship:} Information asymmetry and NPs give rise to blockchain forks and DS attacks (Claim~\ref{claim:DoubleSpendAttack}). Every case where all the nodes do not receive $100\%$ of the information (complete network exposure) correlates to this. It is an opening for {\em DS attacks or censorship}. It can arise due to benign network issues (like downtime and increased connection latencies) or V4 and V5. It is possible in the current Ripple system (RippleNet). Our approaches mitigate this and provide complete network exposure in far worse scenarios. $100\%$ IP to all the nodes in three hops has been demonstrated in Claim~\ref{claim:3hopClaim}, experimentally and in Corollary~\ref{corollary:100infoprop}. It will be reinforced in Simulations. Our approaches provide visibility and reach this objective in far worse scenarios than RippleNet. 

{\flushleft \vspace{-1mm} \bf \ul{(S2, S3 and Solution A.)}} Increased network exposure and speed provides the network's information to every node. It prevents or reduces benign network issues explained earlier. This mitigates benign aspects of V2 and V3. It provides visibility of duplicitous and, or malicious information witholding behaviours  throughout the network and helps in tracking them while obfuscating\footnote{Design choices like not sharing TTL protect connectivity data} sensitive information.

Solution B makes it possible to analyze behaviours for lesser time with lesser errors. Fewer nodes are needed as they receive all or more data vs. earlier cases of partial data receipt. Detecting nodes launching some attacks becomes easier. Some factors that bred uncertainty get removed. Accuracy of flagging becomes easier to prove. Errors in flagging reduce. Chances of a non-malicious node facing benign network issues or being under attack reduce. Malicious nodes can hide lesser behind these issues to escape flagging.

{\flushleft \vspace{-0.75mm} \bf C.) \ul{Information propagation (IP), UNL overlap and consensus} \ul{threshold:}} From Claim~\ref{claim:3hopClaim} if a node can send out a message, then the whole network will receive it in $3$ hops. If the node is unable to do so and is eclipsed, then the network will not receive it. Upon identification of this situation via censorship detection~\cite{censorshipDetectors}, etc. the node can come out of it. Irrespective, the network receives the same information and is in a {\em consistent state}. To prevent eclipse attacks, we suggest efficient random selection from a list of the network's nodes. The UNL selection process can verify another server's community maintained track record~\cite{rippleTechTalk}. We address IP and receipt and its security implications. 

We use the relation between UNL overlap and consensus threshold~\cite{armknecht2015ripple}. The minimum threshold can be reduced with increase in UNL overlap. The consensus time decreases further. We can have $51\%$ threshold for a $98\%$ overlap. \footnote{Minimum $66\%$ threshold is valid in classical byzantine generals problem. We can have $> 50\%$ threshold with public key cryptography.} 

{\flushleft \vspace{-0.75mm} \bf Key arguments connecting overlap with information propagation (IP) and percentage byzantine nodes:} 

\begin{definition} \label{definition:blackboxsame} \vspace{-0.5mm} From the perspective of black box systems, two servers are said to be the same if they always receive the same input and generate the same output.
\end{definition}

{\vspace{-0.75mm} Nodes} take actions or decisions at the start or end of consolidative phases in Ripple and PBFT systems. Comparison between nodes is applicable at these checkpoints. Definition~\ref{definition:blackboxsame} can be loosened and 'same'ness measured at these checkpoints. In Ripple, it corresponds to (1) candidate set generation and (2) each consensus sub-round. This negates the difference between two nodes arising due to timing and ordering of messages and propagation within each consolidative phase. Definition~\ref{definition:blackboxsame} gets loosened to Definition~\ref{definition:practicallySameServers}.

\begin{definition}\label{definition:practicallySameServers} \vspace{-0.5mm} In the context of Ripple and PBFT systems: From a black box perspective and for all practical purposes, two servers are said to be the same if they receive the same input/ information (candidate sets and proposals) and produce the same output/ behave similarly (1. relay the same information and 2. send proposals (yes votes) to the same transactions during consensus sub-round voting) across each consolidative phase.
\end{definition} 

{\flushleft \vspace{-0.75mm} \em Same input:} Our approaches ensure $100\%$ IP  as demonstrated by Claim~\ref{claim:3hopClaim}, Corollary~\ref{corollary:100infoprop} and reinforced experimentally in Simulations. Thus the same information is received at all nodes across consolidative phases.

Note: Two transactions of an account with the same Account Sequence~\cite{ripple:accountSequence} arriving in the same phase is an attack. Then there would be no ordering. Owing to our approaches (Kelips-like overlay and MCA), both the transactions would be identified during candidate set generation and rejected at all nodes. A mechanism to report these can be created.

\begin{definition} \vspace{-0.5mm} Good (non-malicious) nodes are those running accurate implementations of the RPCA as same or different software versions.
\end{definition}

{\flushleft \vspace{-0.75mm} \em Same output:} Two good (non-malicious) nodes implement RPCA accurately and thus exhibit the same behaviour. They receive the same information across consolidative phases using our approaches (including our suggested wait times). They would produce the same output across these phases.

\begin{claim} \label{claim:goodNodeOverlap} \vspace{-0.5mm} Two nodes $\alpha$ and $\beta$ having two good (non-malicious) nodes A and B in their UNLs respectively have an overlap in A and B provided our approaches are used.
\end{claim}
\begin{proof} \vspace{-1.5mm} Assume two good (non-malicious) nodes A and B in UNLs of nodes $\alpha$ and $\beta$ respectively. Given our approach, B receives all messages received by A. Being good, it relays the same over as A. Being good and receiving the same messages as A it sends proposals (aka yes votes) (say, $P1$) to the same transactions (say, $T1$) as A (say, proposal $P2$ for $T1$). For all practical purposes and by Definition~\ref{definition:practicallySameServers}, B is the same as A. Thus $\alpha$ and $\beta$ have an overlap in A and B. 
\end{proof}

{\vspace{-1.5mm} We extrapolate} Claim~\ref{claim:goodNodeOverlap}. All nodes have the latest information. $\alpha$ has $f_{1}$ good nodes in its UNL of size $N_{1}$ and $\beta$ has $f_{2}$ good nodes in its UNL of size $N_{2}$. Overlap between $\alpha$ and $\beta$ is $\geq \min(\frac{f_{1}}{N_{1}},\frac{f_{2}}{N_{2}})$. Extrapolating for the network, the threshold for the absence of forks is $\leq 1-\frac{\min(\frac{f_{1}}{N_{1}},\frac{f_{2}}{N_{2}},\ldots,\frac{f_{n}}{N_{n}})}{2}$. $100 \%$ overlap and thresholds as low as $(50+\epsilon)\%$ ($\epsilon$ is a negligible positive value) is achieved if all nodes are good.

Thus {\em UNL overlap is a function of IP (candidate set and proposal) and the percent of non-byzantine nodes in UNLs.} Increase in propagation with time increases UNL overlap and decreases consensus thresholds~\cite{armknecht2015ripple}. {\em With $100\%$ IP, UNL overlap and consensus thresholds becomes purely a function of the percentage of non-byzantine nodes in UNLs.} 

{\flushleft \vspace{-1mm} \bf Implementation perspective:} A node should not share its UNL data including the nodes believed to be byzantine to prevent Eclipse and other attacks. Thus we can not find overlap and thresholds. This is a challenge. We leverage the fact that a track record of nodes' behaviour~\cite{rippleTechTalk} is expected to be maintained by the community or the node. Nodes can freely choose which records to trust. This is a loosely coupled system. Now a node can compute the 'reputation' or trust value of its UNL and each node in it. For $100\%$ IP, the overlap and thresholds depend only on the UNL's trust value.\footnote{This consideration is at each node as opposed to at the network}. Consensus thresholds can be modified at each node by setting a minimum trust value to be maintained.

{\flushleft \vspace{-1mm} \bf Other merits:} Our approach is more scalable (UNL size: $O(\sqrt{N})$) compared to hard overlap (UNL size: $O(N)$). We systematize UNL selection, improve IP and UNL overlap, making it more scalable. We suggest the use of {\em hard overlap for validation and our concept of overlap for consensus}. 

{\flushleft \vspace{-1mm} \bf Usage perspective:} Depend on the LVL for large transactions and on the LCL or the LVL for small amounts. This depends if one can wait for a few more seconds for the LVL.

\vspace{-2.5mm}
\subsection{Safety and Liveness}
{\vspace{-1mm} We now cover} safety and liveness properties. Ripple demonstrates liveness via its simulator~\cite{simCode}. It shows the system reaching consensus eventually. Our MCA enhances RPCA with IP beyond immediate neighbors. The UNL guidelines further it. There is no major change to the core protocol. Liveness as demonstrated by the simulator~\cite{simCode} holds. IP and consensus and thus liveness is achieved sooner.

We provide adequate overlap for agreement~\cite{armknecht2015ripple}. Agreement ensures safety from NPs, forks and DS Attacks. Malicious behaviours have been identified and solved. In simulations we simulate for safety towards {\bf(1)} correct consensus, and {\bf(2)} $100\%$ IP to all nodes. These are achieved much sooner and in more adverse scenarios with our approach. We ensure greater safety and sooner. {\em Thus we improve both safety and liveness.} 

\vspace{-2.5mm}
\section{Simulation} \label{sec:5}

\vspace{-1mm}
\subsection{Experimental setup}
{\vspace{-0.25mm} We} run a variety of seeded test cases for each simulation on a 256-node network. We build upon Ripple's~\cite{schwartz2014ripple} 'Sim.cpp'~\cite{simCode}. Each case begins with information originating at a random source node. All other nodes are unaware of it. We simulate information propagation {\bf (IP)} which is core to our paper and the base arguments in our security model (Section~\ref{subsec:formalAnalysis}). Propagation of only valid transactions occurs because an invalid transaction is not forwarded by other genuine nodes {\bf (GNs)} and it dies down\footnote{This is valid only for our approach SimK. We extend it to the current Ripple system SimC for comparison.}. A valid transaction is forwarded by all GNs receiving it. {\em Given the scope of our paper (IP) and our security model}, we consider vulnerabilities at a network layer. Thus malicious nodes attempt to throttle or drop transactions. We simulate the dropping of transactions, the worse out of message throttling and dropping. The simulation is at a transaction and proposal level\footnote{We simulate for 'a' transaction. Extendable to 'all' transactions.} Selective, partial or complete throttling of messages is still possible for other transactions. 

Another detail is that $80\%$ of a node's UNL has to agree on the transaction for consensus in SimC and SimRM. This threshold reduces over time (due to IP) for our approach SimK. The simulation verifies whether the system comes to a {\em right consensus} or if all GNs have received the information basis the mode chosen. Several success metrics (not achieving which is failure) are considered and discussed with simulation modes shortly here after.

We simulate for safety, making improvements over 'Sim.cpp'\cite{simCode} which has some flaws including: $\bullet$ The original starting condition of $50\%$ of the network voting in favor and $50\%$ against a transaction is not accurate. It does not represent {\em IP} that influences the starting state.
$\bullet$ State change at each node starts as soon as the node receives votes $=50\%$ of smallest UNL size.
$\bullet$ Consensus threshold is a property of a node. It is wrongly applied to the whole network for success.
$\bullet$ The simulator only proves that the system comes to a 'consensus' and not a right consensus.

We now describe the 3 different versions of the system used for comparison\footnote{[x] represents the mode and is always present. [y] represents percentageEclipsed and is present if applicable}: 
{\bf (1) SimC[x][y]}: Sim-Classic (or SimC) is the current state of RippleNet. It is made using 'Sim.cpp'. UNL size and selection procedure, number of links, link latencies and connections are same as 'Sim.cpp'.
{\bf (2) SimRM[x][y]}: Sim-RippleMid (or SimRM) is a modified version of SimC. UNL size and  number of links is slightly greater than SimK. UNL selection procedure, link latencies and connections is randomised as is in SimC. It represents a randomised topology with UNL size and number of links similar to SimK.
{\bf (3) SimK[x][y]}: Sim-Kelips or (SimK) has Kelips inspired UNL membership and network topology. The UNL, number of links, connections are as per our approach. Link latencies are same as SimC.

SimRM serves as a sanity check showing that the current UNL design is inferior to SimK despite having more links. Actual comparison should be between SimC and SimK.

\pgfplotsset{scale=0.5}
\begin{figure*}[!htb]
\centering
\includegraphics{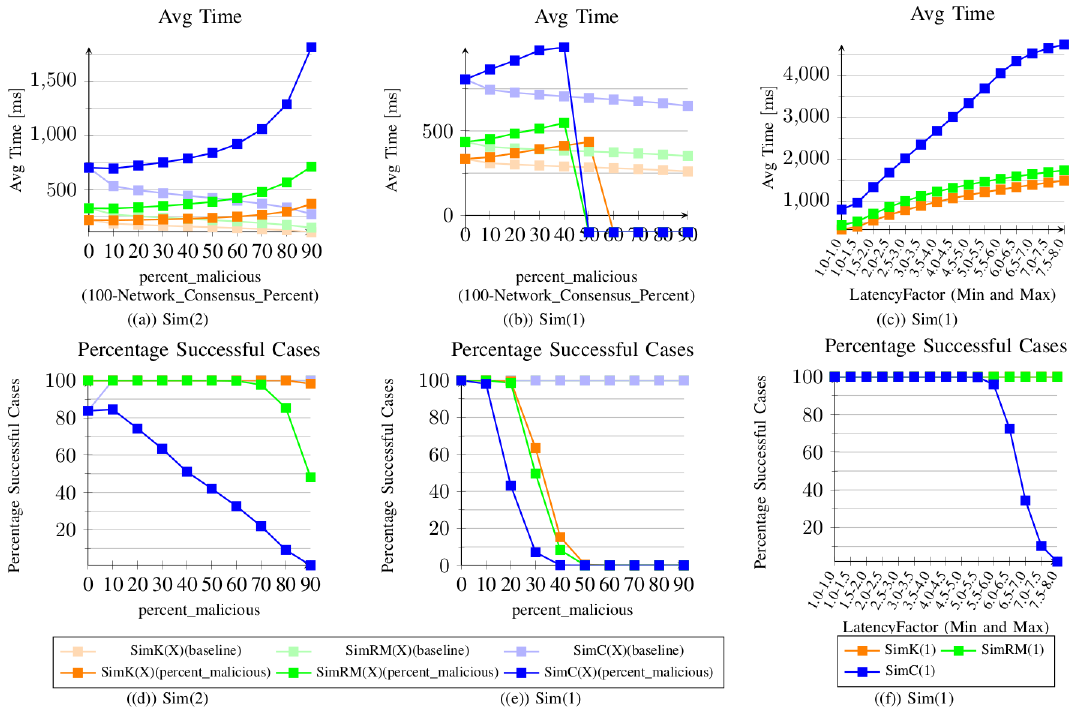}
\caption{\label{fig:degreeSeveritySetup}
((a)), ((d)) compare IP for base cases ($PM=0$) and PM cases in the absence of network issues for mode 2. 
((b)), ((e)) compares consensus for base cases ($PM=0$) and PM cases in the absence of network issues for mode 1.
((c)), ((f)) compare consensus for $PM=0$ , $PLANI=75$, $PNANI=100$, $NCP=100$ and mode 1.}
\vspace{-6mm}
\end{figure*}

\pgfplotsset{scale=0.5}
\begin{figure*}[!htb]
\centering
\subfloat[Sim(2)]{
\centering
\begin{tikzpicture}
\begin{axis}[
axis lines=left, 
legend style={
                at={(xlabel.south)}, % Place legend relative to xlabel node
                yshift=-1ex,
                anchor=north,
                font = \scriptsize
            },
ymajorgrids=true,
yminorgrids=true,
minor y tick num = 1,
title={Avg Time},
y SI prefix=milli,y unit=s,  
  xlabel= Severity,
  xlabel style = {font = \footnotesize, name=xlabel},
  ylabel=Avg Time,
  ylabel style = {font = \footnotesize, align=center},
  enlargelimits = false,
  xticklabels from table={Sim2-baselineSeverityComparison}{SeverityNum},xtick=data]
\addplot[colorLight3Orange,thick,mark=square*] table [y=SimK(2)(baseline)AvTime,x=SeverityNum]{Sim2-baselineSeverityComparison};
\addplot[colorLight3Green,thick,mark=square*] table [y=SimRM(2)(baseline)AvTime,x=SeverityNum]{Sim2-baselineSeverityComparison};
\addplot[colorLight3Blue,thick,mark=square*] table [y=SimC(2)(baseline)AvTime,x=SeverityNum
]{Sim2-baselineSeverityComparison};

\addplot[orange,thick,mark=square*] table [y=SimK(2)(severity)AvTime,x=SeverityNum]{Sim2-baselineSeverityComparison};
\addplot[green,thick,mark=square*] table [y=SimRM(2)(severity)AvTime,x=SeverityNum]{Sim2-baselineSeverityComparison};
\addplot[blue,thick,mark=square*] table [y=SimC(2)(severity)AvTime,x=SeverityNum
]{Sim2-baselineSeverityComparison};
\end{axis}
\end{tikzpicture}
\label{fig:III2466}}
\hfill
\subfloat[Sim(4) vs Sim(6)]{
\centering
\begin{tikzpicture}
\begin{axis}[
axis lines=left, 
legend style={
                at={(xlabel.south)}, % Place legend relative to xlabel node
                yshift=-1ex,
                anchor=north,
                font = \scriptsize
            },
ymajorgrids=true,
yminorgrids=true,
minor y tick num = 1,
title={Avg Time},
y SI prefix=milli,y unit=s,  
  xlabel= Severity,
  xlabel style = {font = \footnotesize, name=xlabel},
  ylabel=Avg Time,
  ylabel style = {font = \footnotesize, align=center},
  enlargelimits = false,
  xticklabels from table={Sim4vsSim6}{SeverityNum},xtick=data]
\addplot[colorLight3Orange,thick,mark=square*] table [y=SimK(4)AvTime,x=SeverityNum]{Sim4vsSim6};
\addplot[colorLight3Green,thick,mark=square*] table [y=SimRM(4)AvTime,x=SeverityNum]{Sim4vsSim6};
\addplot[colorLight3Blue,thick,mark=square*] table [y=SimC(4)AvTime,x=SeverityNum]{Sim4vsSim6};

\addplot[orange,thick,mark=square*] table [y=SimK(6)(100)AvTime,x=SeverityNum]{Sim4vsSim6};
\addplot[green,thick,mark=square*] table [y=SimRM(6)(100)AvTime,x=SeverityNum]{Sim4vsSim6};
\addplot[blue,thick,mark=square*] table [y=SimC(6)(100)AvTime,x=SeverityNum
]{Sim4vsSim6};
\end{axis}
\end{tikzpicture}
\label{fig:III2466}}
\hfill
\subfloat[Sim(6)]{
\centering
\begin{tikzpicture}
\begin{axis}[
axis lines=left, 
legend style={
                at={(xlabel.south)}, % Place legend relative to xlabel node
                yshift=-4ex,
                anchor=north,
                font = \scriptsize
            },
ymajorgrids=true,
yminorgrids=true,
minor y tick num = 1,
title={Avg Time},
y SI prefix=milli,y unit=s,  
  xlabel= Percentage Eclipsed,
  xlabel style = {font = \footnotesize, name=xlabel},
  ylabel=Avg Time,
  ylabel style = {font = \footnotesize, align=center},
  enlargelimits = false,
  xticklabels from table={Sim6-percentEclipsed}{percentageEclipsed},xtick=data]
\addplot[orange,thick,mark=square*] table [y=SimK(6)AvTime,x=percentageEclipsed]{Sim6-percentEclipsed};
\addplot[green,thick,mark=square*] table [y=SimRM(6)AvTime,x=percentageEclipsed]{Sim6-percentEclipsed};
\addplot[blue,thick,mark=square*] table [y=SimC(6)AvTime,x=percentageEclipsed]{Sim6-percentEclipsed};
\end{axis}
\end{tikzpicture}
\label{fig:III2466}}
\hfill
\subfloat[Sim(2)]{
\centering
\begin{tikzpicture}
\begin{axis}[
axis lines=left, 
legend style={
                at={(xlabel.south)}, % Place legend relative to xlabel node
                yshift=-1ex,
                anchor=north,
                font = \scriptsize
            },
ymajorgrids=true,
yminorgrids=true,
minor y tick num = 1,
title={Percentage Successful Cases},
  xlabel=Severity,
  xlabel style = {font = \footnotesize, name=xlabel},
  ylabel=Percentage Successful Cases,
  ylabel style = {font = \footnotesize, align=center},
  enlargelimits = false,
  xticklabels from table={Sim2-baselineSeverityComparison}{SeverityNum},xtick=data,
  legend columns=2]
\addplot[colorLight3Orange,thick,mark=square*] table [y=SimK(2)(baseline)PSC,x=SeverityNum]{Sim2-baselineSeverityComparison};
\addlegendentry{SimK(2)(base)}
\addplot[colorLight3Green,thick,mark=square*] table [y=SimRM(2)(baseline)PSC,x=SeverityNum]{Sim2-baselineSeverityComparison};
\addlegendentry{SimRM(2)(base)}
\addplot[colorLight3Blue,thick,mark=square*] table [y=SimC(2)(baseline)PSC,x=SeverityNum]{Sim2-baselineSeverityComparison};
\addlegendentry{SimC(2)(base)}

\addplot[orange,thick,mark=square*] table [y=SimK(2)(severity)PSC,x=SeverityNum]{Sim2-baselineSeverityComparison};
\addlegendentry{SimK(2)(sev)}
\addplot[green,thick,mark=square*] table [y=SimRM(2)(severity)PSC,x=SeverityNum]{Sim2-baselineSeverityComparison};
\addlegendentry{SimRM(2)(sev)}
\addplot[blue,thick,mark=square*] table [y=SimC(2)(severity)PSC,x=SeverityNum]{Sim2-baselineSeverityComparison};
\addlegendentry{SimC(2)(sev)}
\end{axis}
\end{tikzpicture}
\label{fig:III2466}}
\hfill
\subfloat[Sim(4) vs Sim(6)]{
\centering
\begin{tikzpicture}
\begin{axis}[
axis lines=left, 
legend style={
                at={(xlabel.south)}, % Place legend relative to xlabel node
                yshift=-1ex,
                anchor=north,
                font = \scriptsize
            },
ymajorgrids=true,
yminorgrids=true,
minor y tick num = 1,
title={Percentage Successful Cases},
  xlabel=Severity,
  xlabel style = {font = \footnotesize, name=xlabel},
  ylabel=Percentage Successful Cases,
  ylabel style = {font = \footnotesize, align=center},
  enlargelimits = false,
  xticklabels from table={Sim4vsSim6}{SeverityNum},xtick=data,
  legend columns=2]
\addplot[colorLight3Orange,thick,mark=square*] table [y=SimK(4)PSC,x=SeverityNum]{Sim4vsSim6};
\addlegendentry{SimK(4)}
\addplot[colorLight3Green,thick,mark=square*] table [y=SimRM(4)PSC,x=SeverityNum]{Sim4vsSim6};
\addlegendentry{SimRM(4)}
\addplot[colorLight3Blue,thick,mark=square*] table [y=SimC(4)PSC,x=SeverityNum]{Sim4vsSim6};
\addlegendentry{SimC(4)}

\addplot[orange,thick,mark=square*] table [y=SimK(6)(100)PSC,x=SeverityNum]{Sim4vsSim6};
\addlegendentry{SimK(6)(100)}
\addplot[green,thick,mark=square*] table [y=SimRM(6)(100)PSC,x=SeverityNum]{Sim4vsSim6};
\addlegendentry{SimRM(6)(100)}
\addplot[blue,thick,mark=square*] table [y=SimC(6)(100)PSC,x=SeverityNum]{Sim4vsSim6};
\addlegendentry{SimC(6)(100)}
\end{axis}
\end{tikzpicture}
\label{fig:III2466}}
\hfill
\subfloat[Sim(6)]{
\centering
\begin{tikzpicture}
\begin{axis}[
axis lines=left, 
legend style={
                at={(xlabel.south)}, % Place legend relative to xlabel node
                yshift=-1ex,
                anchor=north,
                font = \scriptsize
            },
ymajorgrids=true,
yminorgrids=true,
minor y tick num = 1,
title={Percentage Successful Cases},
  xlabel=Percentage Eclipsed,
  xlabel style = {font = \footnotesize, name=xlabel},
  ylabel=Percentage Successful Cases,
  ylabel style = {font = \footnotesize, align=center},
  enlargelimits = false,
  xticklabels from table={Sim6-percentEclipsed}{percentageEclipsed},xtick=data]
\addplot[orange,thick,mark=square*] table [y=SimK(6)PSC,x=percentageEclipsed]{Sim6-percentEclipsed};
\addlegendentry{SimK(6)}
\addplot[green,thick,mark=square*] table [y=SimRM(6)PSC,x=percentageEclipsed]{Sim6-percentEclipsed};
\addlegendentry{SimRM(6)}
\addplot[blue,thick,mark=square*] table [y=SimC(6)PSC,x=percentageEclipsed]{Sim6-percentEclipsed};
\addlegendentry{SimC(6)}
\end{axis}
\end{tikzpicture}
\label{fig:III2466}}
\caption{\label{fig:III2466}  ((a)), ((d)) compare IP for base cases ($PM=0$ , Ideal condition, NCP corresponding to comparison cases' NCP) and cases of corresponding severity for mode 2.
((b)), ((e)) compare for modes 4 and 6 across varying severity.
((c)), ((f)): compare for mode 6 across varying PE for Real World severity.}
\vspace{-6mm}
\end{figure*}
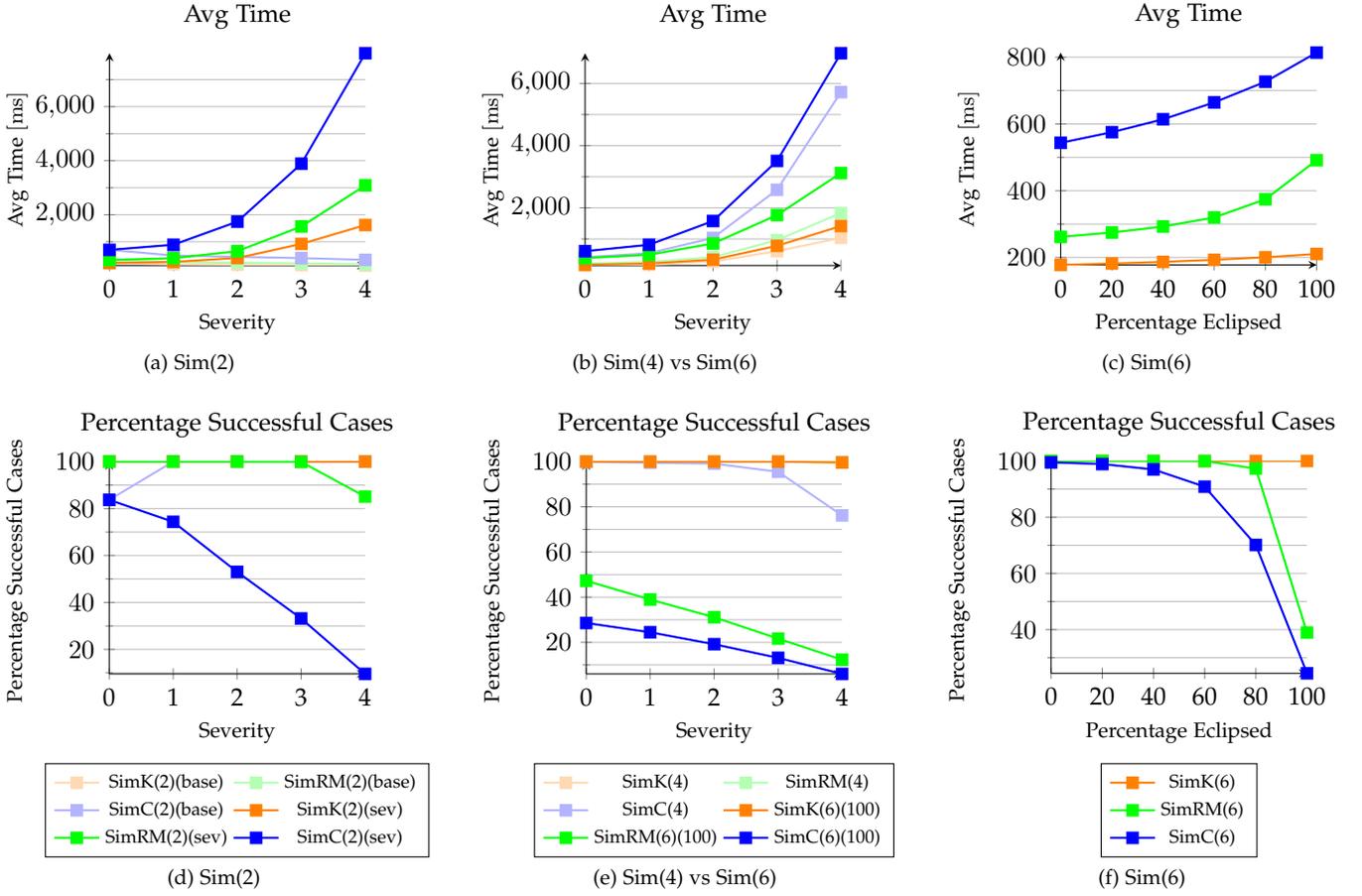

We have various configurables across different simulations:
$\bullet$ \textbf{mode}: The mode in which simulation is being run. The different modes are: 
{(a) \bf mode 1} measures consensus at a network level. Success is achieved when the no. of nodes arriving at the right consensus is $\geq$ min(GNs, no. of nodes corresponding to Network\_Consensus\_Percent [aka NCP]).
{(b) \bf mode 2}  measures IP at a network level. Success is achieved when the no. of nodes that have received information $\geq$ min(GNs, no. of nodes corresponding to NCP).
{(c) \bf mode 3}  measures consensus at a network level and at a random target node. Success metrics include those for mode 1 and the target node coming to consensus.
{(d) \bf mode 4} measures IP to a random target node. 
{(e) \bf mode 5} extends mode 3 with the target node partially eclipsed by malicious nodes in its UNL. Success metric same as mode 3.
{(f) \bf mode 6} extends mode 4 with the target node partially eclipsed by malicious nodes in its UNL.
{(g) \bf mode 7} measures the maximum of shortest distance (MSD) nodes have from the source node receiving the transaction, when a particular target node is partially eclipsed by malicious nodes upto 'percentageEclipsed' nodes in its links (connections). $Success_1$ is achieved when all the GNs have received information by the time  message passing has ended. $Success_2$ is achieved when all the GNs have received information by the time message passing has ended and MSD $\leq 3$.
$\bullet$ \textbf{Network\_Consensus\_Percent (NCP)}: The network wide percentage which is one of the factors of a  minimum function to be achieved for achieving success. Not valid for modes 4, 6 and 7.
$\bullet$ \textbf{percentage\_malicious (PM)}: Percentage of nodes in the network that are malicious and distributed in a random fashion. It does not include nodes eclipsing target nodes.
$\bullet$ \textbf{outbound\_links\_to\_node\_ratio (OLTNR)}: Ratio representing the number of outbound links per node to the number of nodes. Applicable only for SimC. Unless otherwise mentioned,  we consider it to be $10/256$ for simulations. This is higher than the current\footnote{Software version: rippled-1.5.0-rc3} ratio $15/1024$ in a non-randomised topology. As the OLTNR decreases, the graph becomes sparser, and the performance of SimC reduces (as the quality of IP decreases). The chances of failure in consensus increases. Thus, the actual performance of RippleNet in a randomised topology is worse than the performance in simulations. 
% Therefore, the actual benefits of our approach {\color{red} compared to a randomised RippleNet topology} is more than that seen in the experiments. 
$\bullet$ \textbf{minLatencyFactorForNI (minLFNI) and maxLatencyFactorForNI (maxLFNI)}: Minimum (inclusive) and maximum (exclusive) factors  affecting latencies of links affected by Network Issues. The factors are over a uniform real number distribution. 
$\bullet$ \textbf{percentNodesAffectedByNI (PNANI)}: Percentage of nodes affected by Network Issues.
$\bullet$ \textbf{percentLinksAffectedByNI (PLANI)}: Percentage of links affected for each node affected by Network Issues. 
$\bullet$ \textbf{percentageEclipsed (PE)}: The percentage of malicious nodes on the other side of links or in the UNL\footnote{Eclipse node via Malicious UNL for modes 5 and 6, and Eclipse node via Malicious Links for mode 7} (whichever is applicable) of the test node. (Applicable for modes 5, 6 and 7)
$\bullet$ \textbf{seedMax}: The number of seeded cases being run. Unless otherwise mentioned seedMax is $1500$ for odd modes and $5000$ for even modes.  
$\bullet$ \textbf{UNL\_B\_PER\_AFFINITY\_SUBGROUP\_SIZE}: Variable c in UNL-B size as in our paper. For simulations $c=2$ .

{\flushleft \vspace{-0.75mm} \bf Note:} While interpreting the graphs, all values have a natural and non-negative expected range. {\bf The presence of negative value for any particular data point is an indicator that all test cases associated with that data point have ended in failure.} Purposefully incorporated, it shows a distinction between low or 0 values and an all failure case. 

Also, malicious nodes do not add the transactions under consideration in their LCL and do not declare consensus on them in our simulations. Only GNs come to a consensus. If NCP $>$ actual $\%$ GN, NCP is set as the actual $\%$ GN.

It is possible that the percentage of malicious nodes in UNL are $>$ PM nodes in the network. As conditions get severe there can be breakages in the network.  Malicious node configurations may become such that a GN gets $100\%$ eclipsed for IP, or $>50\%$ for consensus or some other issues. In such a situation it may not be possible to have IP and consensus at $100\%$ of the nodes.  Due to these constraints, SimK (though better than SimC) is not able to achieve $100\%$ success in consensus in extremely poor conditions.  

\vspace{-2.5mm}
\subsection{Results and interpretations} \label{Insights1}

\pgfplotsset{scale=0.55}
\begin{figure*}[!htb]
\centering
\includegraphics{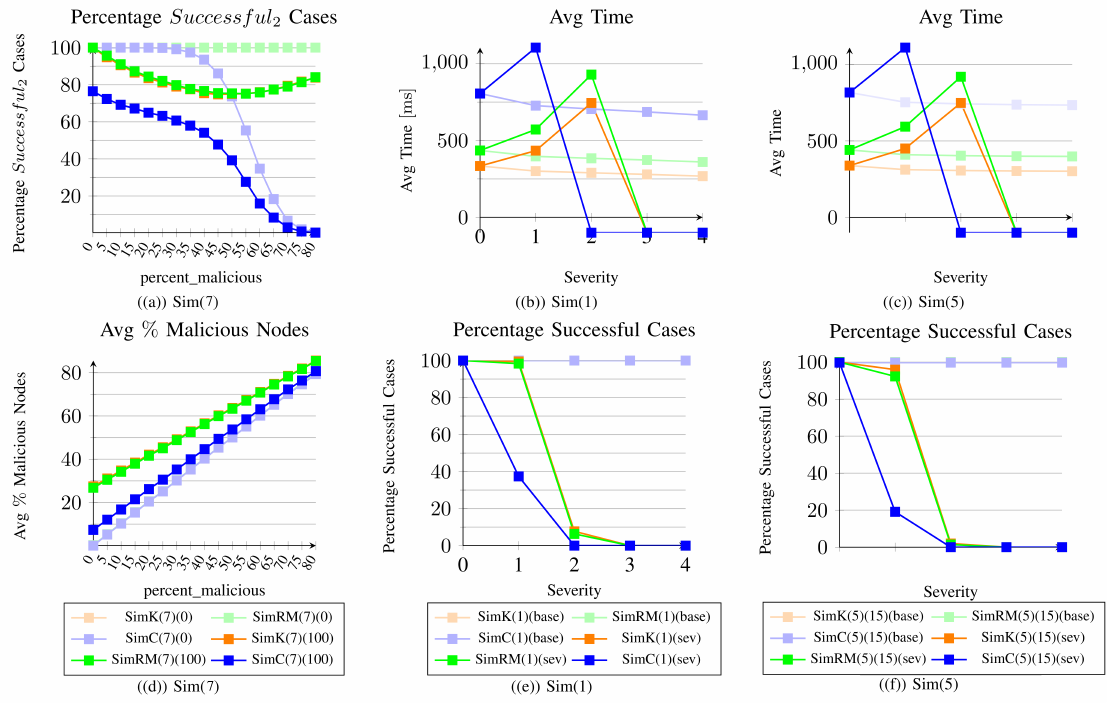}
\caption{\label{fig:ICC815}((a)), ((d)) compare IP for mode 7.
((b)), ((e)) compare consensus for base cases ($PM=0$, Ideal condition, NCP corresponding to comparison cases' NCP) and cases of corresponding severity for mode 1.
((c)), ((f)): compare consensus for base cases ($PM=0$, Ideal condition, NCP corresponding to comparison cases' NCP) and cases of corresponding severity for mode 5(15).}
\vspace{-6mm}
\end{figure*}

\begin{table}[!tb]
% \vspace{-3mm}
\caption{\label{table:severityFactors} Factors associated with varying degrees of severity}
\centering{}%
\begin{tabular}{|c|c|c|c|c|c|}
\hline 
\begin{tabular}{@{}c@{}}Severity\\(Severity Num)\end{tabular}
& \begin{tabular}{@{}c@{}}\\Ideal\\(0)\end{tabular}
& \begin{tabular}{@{}c@{}}Real\\World\\(1)\end{tabular}
& \begin{tabular}{@{}c@{}}\\Mild\\(2)\end{tabular}
& \begin{tabular}{@{}c@{}}Moderate-\\Severe\\(3)\end{tabular}
& \begin{tabular}{@{}c@{}}Very\\Severe\\(4)\end{tabular}  
\tabularnewline
\hline 
\begin{tabular}{@{}c@{}} $\%$ Malicious\\ Nodes\end{tabular}  & $0$ & $20$  & $40$ & $60$ & $80$ \tabularnewline
\hline
\begin{tabular}{@{}c@{}}minLFNI\end{tabular}  & $0$ & $1.5$  & $3$ & $4.5$ & $6.5$ \tabularnewline
\hline
\begin{tabular}{@{}c@{}}maxLFNI\end{tabular}  & $0$ & $2$  & $3.5$ & $5$ & $7$ \tabularnewline
\hline
\begin{tabular}{@{}c@{}}PLANI\end{tabular}  & $0$ & $25$  & $50$ & $75$ & $75$ \tabularnewline
\hline
\begin{tabular}{@{}c@{}}PNANI\end{tabular}  & $100$ & $100$ & $100$ & $100$ & $100$ \tabularnewline
\hline
\begin{tabular}{@{}c@{}}NCP ($\%$GNs)\end{tabular}  & $100$ & $80$ & $60$ & $40$ & $20$ \tabularnewline
\hline
\end{tabular}
\vspace{-6mm}
\end{table}

{\flushleft \vspace{-0.75mm} \bf Degrees of severity:} Degrees of severity are chosen basis simulations shown in Figure~\ref{fig:degreeSeveritySetup}, observations of real world statistics~\cite{wonderNetworkPings} and are tabulated in Table~\ref{table:severityFactors}.

{\flushleft \vspace{-0.75mm} \bf Information Propagation:} Figures \ref{fig:degreeSeveritySetup}, \ref{fig:III2466} provide a comparative view of average time and percentage successful cases. They cover modes 2, 4 and 6 associated with IP, under different cases. 
{\em SimK consistently performs better than SimRM and SimC}. 
By the time 20$\%$ of RippleNet receives information 100$\%$ of our system receives information.
Every case of IP delay or failure, is an opportunity for creating a fork, or for censoring a transaction. {\em Thus, SimK improves security along with transaction speed.}

Figure~\ref{fig:ICC815} provides a comparative view of percentage $successful_2$ cases and average percentage malicious nodes with  mode 7. SimK performs significantly better than SimC. The PM nodes is very slightly higher for SimK as compared to SimRM. Thus SimK's performance is very slightly poor. However, since there are several possible and faster paths in SimK, it performs better as demonstrated in other simulations associated with IP.

{\flushleft \vspace{-0.75mm} \bf Consensus:} Figures \ref{fig:degreeSeveritySetup}, \ref{fig:ICC815} provide a comparative view of average time and percentage successful cases. They cover modes 1, 3 and 5 associated with consensus, under different cases. SimK consistently performs better than SimRM and SimC. Every case of consensus failure is a possibility of genuine transactions getting stuck and forcing failed retries. This consumes resources and reduces the utility and adoptability of the system. There is lesser possibility of consensus failure in our approach SimK over SimC. Thus there are lesser opportunities for an attacker to censor or affect the system.

{\flushleft \vspace{-0.75mm} \bf Protection against Double Spend (DS) Attacks and Censorship:} Figure \ref{fig:degreeSeveritySetup} (d) provides a comparative view for IP. Basis Claim~\ref{claim:DoubleSpendAttack} we assume that cases where $100\%$ IP is not achieved as cases where attackers can effect DS Attacks or censor the network. We look at figure \ref{fig:degreeSeveritySetup} (d) and combine the results to see a cumulative result. It is observed that there are openings for attackers to launch DS Attacks or censor the system a total $53.62\%$, $6.92\%$ and $0.178\%$ times out of $50000$ simulated cases each while varying the configurations and PM nodes for SimC, SimRM and SimK respectively. Given that we have given performance boosts to SimC and SimRM in simulations over the existing systems, they are expected to perform far worse. {\em Thus we get $\geq 99.67\%$ reduction in opportunities for DS attacks and censorship.}

{\flushleft \vspace{-0.75mm} \bf Tolerance to faults and malicious nodes:}
The edge case (current limits or lower bounds) of the existing Ripple system is at $20\%$ malicious nodes. The same success rate of consensus as observed at this edge case is observed at $34.21\%$ malicious nodes using our approaches. As before we account for extra performance boost provided to SimC. {\em This implies $1.71x$ increase in fault tolerance from $20\%$ to $\geq 34.21\%$.}

\section{Limitations and Future Work}\label{sec:futureWork}
$\bullet$ A methodology to build and maintain a dynamic UNL while avoiding and preventing various attack vectors via multiple avenues including randomisation and membership algorithms for UNLs, detection is out of scope of this paper and the subject of our ongoing research.

$\bullet$ A methodology to calculate and maintain the trust value of UNLs at a node and network level, and reputation management in a loosely coupled system is out of scope of this paper and the subject of our ongoing research.

$\bullet$ A methodology to identify if a malicious node/ account is sending two different sets of messages to different nodes is the subject of our ongoing research.

$\bullet$ Optimisations to reduce chattiness in the context of our modified consensus algorithm can be further looked upon while implementing it and are not covered.

\section{Conclusion} \label{sec:conclusion}
Our approaches and design are inherently generic and are applicable to other cryptocurrencies and systems. Our paper explores and opens concepts associated with systematising the UNL and its overlap, implementation of efficient message propagation and consensus. It further improves the RPCA, and introduces consensus thresholds as a function of information propagation and percentage of non-byzantine nodes in the UNL. It brings in protection against attacks and vulnerabilities. It thus has potential to bring about benefits such as reduced consensus thresholds, increased speed, improved and systematized security, resilience to Double Spend Attacks, Censorship and other byzantine attacks, last mile connectivity. It assures 2 degrees of separation (3 hops). These improvements ensure promotion of a more energy efficient consensus over Proof of Work.

% \clearpage
% \vfill

\printbibliography
\clearpage

\appendices

\section{Abbreviations}

\subsection{Frequently used abbreviations:}

\begin{itemize}
    \item BGP: Byzantine General's Problem
    \item DS: Double Spend
    \item IP: Information Propagation
    \item MCA: Modified Consensus Algorithm
    \item NP: Network Partition
    \item RPCA: Ripple Protocol Consensus Algorithm
    \item UNL: Unique Node List
\end{itemize}

\subsection{Infrequently used abbreviations}

\begin{itemize}
    \item DHT: Distributed Hash Table
    \item DNS: Domain Name Server
    \item P2P: Peer-to-peer
    \item TNL: Trustee Node List
    
\end{itemize}

\subsection{Abbreviations in simulations:}

\begin{itemize}
    \item GN: Genuine Node
    \item maxLFNI: maxLatencyFactorForNI
    \item minLFNI: minLatencyFactorForNI
    \item MSD: maximum of shortest distance
    \item NCP: Network\_Consensus\_Percent
    \item OLTNR: outbound\_links\_to\_node\_ratio
    \item PE: percentageEclipsed
    \item PM: percentage\_malicious
    \item PLANI: percentLinksAffectedByNI
    \item PNANI: percentNodesAffectedByNI
\end{itemize} 

\section{Goals and Challenges of Distributed Payment Systems} \label{sec:2}
Owing to their distributed nature, blockchains also suffer from the challenge of being able to achieve only two of the three metrics of an impossible trinity of scale, security, and decentralization (derivative of the CAP theorem~\cite{brewer2000CAP}).

In this paper we work to improve scale and security at no or negligible cost to decentralisation. While the introduction of introducers negligibly or not at all affect decentralisation in a negative fashion, the automated UNL ensures decentralisation whilst affecting it positively. 

{\flushleft \bf Goals of Distributed Payments Systems}
Some of the goals and areas of work in RTGSs, Distributed Payment Systems and allied systems (in general), needed towards their widespread adoption and success include improvement of utility by 
\begin{enumerate}[leftmargin=*]
    \item improving the transaction speed without compromising on security (with the same or similar levels of certainty) 
    \item having in place highly resilient and provably secure systems with reasonably comfortable consensus and safety thresholds that are not prone to network partitioning and blockchain forks.
    \item improving last mile connectivity (thus being able to conduct transactions even in relatively poor connectivity areas).
\end{enumerate}

{\flushleft \bf The Hazards of Network Partitions and Blockchain Forks}
Network partitions and blockchain forks are detrimental to  systems because forks detected and resolved post the completion of  transactions open the system to Double Spend attacks (illegitimate transactions). It is  often resolved by hard forks with legitimate transactions on the discarded forks being invalidated.

Failure to achieve the consensus threshold's lower bounds results in invalidation and rejection of the transaction. Some of the reasons for this failure include:
        \begin{enumerate}
            \item Not enough trusted members participating in consensus agreeing (voting positively) on the transaction as they did not deem it valid.
            \item Non-receipt of transactions due to network issues (latency, DDoS and Eclipse attacks, flux and churn, message throttling, etc), which can lead to network partitioning, blockchain forks and double spend attacks.
        \end{enumerate}
Rejection of transactions and thus prevention of the blockchain's forward movement in a potentially unhealthy network situation prevents potential blockchain forks and double spend attacks. However this contains a flaw -- forward movement is prevented merely on the possibility or suspicion of a network partition and not just when the network truly undergoes partitioning. It thus affects the normal functioning and brings to a halt all legitimate transactions and financial activities, dependent systems which have adopted and integrated with it, partially or fully. Thus it has a negative impact on real life human activities and though relatively safer, is still a loss making situation in itself.

% \todo[inline, color=brown]{This part below wrt public key cryptography can possibly be made a footnote. However I am unable to find an point where it will fit as a footnote. Should it even be a footnote?}
% \\Implementing public key cryptography prevents man in the middle attacks by ensuring that only the legitimate sender is able to send the message and the message is untampered with, thus eliminating one of the major concerns of the byzantine generals problem. 
% \begin{enumerate}
%     \item A node can still behave maliciously by not forwarding messages and thus preventing the other nodes from making an appropriate choice.
%     \item While the nodes cannot create faulty transactions, double spend attacks by malicious clients by issuing two transactions with the same transaction number to two different nodes is still a possibility.
% \end{enumerate}
% While public key cryptography makes it possible to identify the malicious account (not necessarily the malicious user), the damage is already done. The money may not always be recoverable. 

The above mentioned challenges undermine systems open to them by undermining their reliability and the trust in them, questioning their very veracity. They raise adoption and viability concerns causing potential loss of new and existing members and reduction in transaction volume. One should thus be mindful of the above challenges and various security considerations and also ensure that normal operations and activity is not impacted  by false positives, while designing systems. Engineering and research work going on to make systems secure, fast, fault and partition tolerant thus is of importance and needed to ensure trust and widespread adoption. 

{\flushleft \bf Solving for Distributed Payments Systems, and in general}
As one can see, the absence of an appropriate consensus process and the possibility of network partitioning or a delay in information propagation can be extremely detrimental. There is a need for 
\begin{itemize}
\item a consensus algorithm that prevents such vulnerabilities
\item an information propagation mechanism that ensures malicious nodes are  unable to prevent the network from making forward progress and unable to cause network partitioning and forks in the system.
\end{itemize}

In our work, we have tried to fulfill these needs and gaps and worked towards achieving the goals associated with distributed payments systems mentioned above. We present guidelines to implement and set up Ripple's UNL, introduce newer lists and constructs, setting a base network overlay structure for assured overlap and efficient information propagation (in generic and randomised network topology). We also outline various node behaviours, present a modified consensus algorithm. We back this with sound analysis and experimental data and bring about improvements as outlined in the previous section and also later in the paper, including in areas such as performance, security, utility, thereby impacting adoption. We also provide methods to automate UNL updation.
\iffalse Our work ensures that-
\begin{enumerate}
    \item consensus is attained at lower thresholds and in reduced time, while ensuring the same (if not more) degree of trust, certainty and guarantee in transactions as was with the original system. 
    \item last mile connectivity is provided. 
    \item there is a reduction in chances of network partitions and forks, and the system is more secure.
    \item the overlay network can be utilised to determine the network state, latency and for connection detection.
    \item UNL updation is automated, and UNL generation and the system is randomised and not vulnerable to attacks and secured by approaches such as the ones mentioned in BRAHMS paper  
\end{enumerate}
\fi

It may be noted that since the network topology on top of which the UNL and other constructs is implemented is generic per se, the UNL and other construct's implementation too is generic. Our work is applicable not just to Ripple, a representative cryptocurrency and distributed payments system but to all other RTGSs, distrbuted payments systems, cryptocurrencies and to some problems faced as a part of the Byzantine Generals Problem

\section{Modified Consensus Algorithm}

The duration $y$ for Sub-round C of Candidate Set Generation phase (discussed earlier in the paper) includes time for (1) min 3 hops\footnote{In this paper, we imply 'hops' to be vertex hops.} of IP in our overlay and (2) distributed nature of consensus (factor of 2). Thus $y \geq 6 \times x$.

\section{The mechanics and percentage of information propagation} \label{MechAndPercentInfoProp}
As part of this subsection, we represent assured propagation routes from a source node to destination node in the same affinity group. More routes, similar or otherwise might be possible. For source and destination nodes in separate affinity groups, a similar approach can be followed.

As shown in figure~\ref{fig:1Hop}, the information propagates to all nodes in the source node's UNL. $\frac{(c+1)\times (\sqrt{N}-1) \times 100}{N} \%$ of the network receives the information as a best case scenario, including the destination node. \footnote{These $(c+1)\times (\sqrt{N}-1)$ nodes are the single biggest vulnerabilities for dropping messages or throttling any information to and fro nodes. As long as information reaches any one of these nodes, information is likely to reach the rest of the network as shown later}

\begin{figure}[!h]
  \centering
  \includegraphics[scale=0.4]{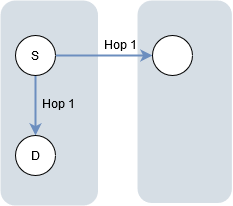}
  \caption{\label{fig:1Hop}Propagation in 1 Hop: Source and Destination in Same Affinity Group}
\end{figure}

\begin{figure}[H]
  \centering
  \includegraphics[scale=0.3]{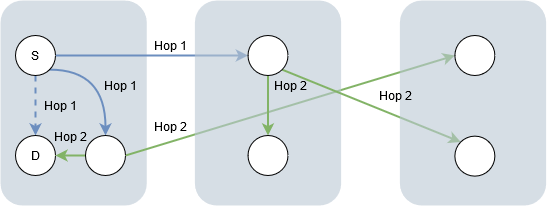}
  \caption{\label{fig:2Hop}Propagation in 2 Hops: Source and Destination in Same Affinity Group}
\end{figure}

In case the destination node does not receive the information directly in one hop, it receives the same via two hops through all other nodes in its affinity group, as shown in figure \ref{fig:2Hop}. These are assured mechanisms. It might be possible that the source and destination node share common nodes as connections in other affinity groups. As a best case scenario, $100 \%$ of nodes in the network receive the information within two hops. 

\begin{figure}[H]
  \centering
  \includegraphics[scale=0.3]{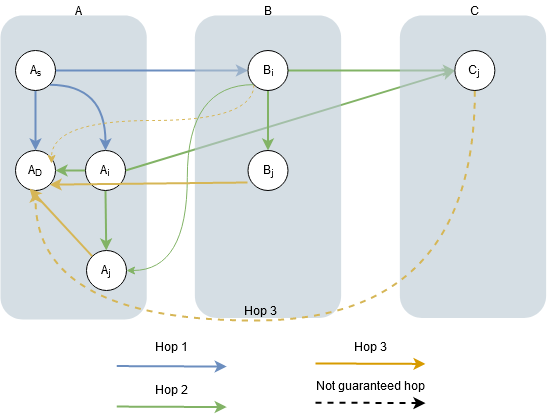}
  \caption{\label{fig:SandDSameAffinityGroup}Propagation in 3 Hops: Source and Destination in Same Affinity Group}
\end{figure}

In case the destination node still does not receive information within two hops, it receives the same in three hops via the routes depicted in the figure \ref{fig:SandDSameAffinityGroup} and other routes. It is guaranteed that $100 \%$ of the network receives the information in the absence of $\geq (c+1) \times (\sqrt{N}-1)$ failures.

\begin{figure}[H]
  \centering
  \includegraphics[scale=0.3]{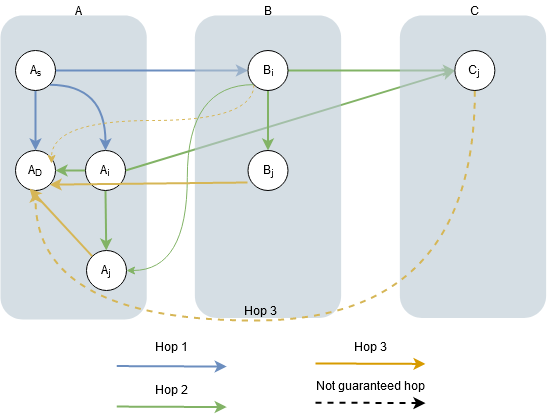}
  \caption{\label{fig:SandDSeparateAffinityGroup}Source and Destination in Separate Affinity Groups}
\end{figure}

This ensures a distance of $3$ hops (maximum $2$ degrees of separation) between any two nodes in the absence of $\geq (c+1) \times (\sqrt{N}-1)$ failures. For the presence of $\geq (c+1) \times (\sqrt{N}-1)$ failures, a node is either uneclipsed, partially eclipsed, fully eclipsed. If a node is uneclipsed or partially eclipsed, it is still able to send and receive all information from the network including the last validated ledger. If a node is fully eclipsed and it is no longer a part of the network and not participating in the consensus process and not associated with any genuine nodes. A node admin should identify such a situation (existing means are censorship detectors~\cite{censorshipDetectors}) and rectify it.

\section{Simulation}
{\flushleft \bf Other configurables} other than those mentioned earlier include: 
$\bullet$ \textbf{UNLA\_LLF\_MAX, UNLB\_LLF\_MAX}: The maximum link latency factor\footnote{Latencies of links and thus the ratio of link latencies for links within and outside affinity groups can be varied by a factor equal to the link latency factor. Specific link latency factors are represented by $\{k,l\}$ with $k$ and $l$ being link latency factor for UNLA and UNLB respectively. Applicable for SimRB and SimK. A similar nomenclature if used for SimC is applicable only $k=l$.}(LLF) for links associated with UNLA and UNLB for SimK and a similar association for links in SimRB. All associated permutations computed using these levers.
$\bullet$ \textbf{upperLimitMalicious}: The upper limit (exclusive) of the number of malicious nodes in the system. Corresponds to $(c+1)\times(\sqrt{N}-1)$. Only applicable for mode 7.
$\bullet$ \textbf{isUpperLimitMaliciousApplicable}: Boolean value ensuring applicability of upperLimitMalicious. Only applicable for mode = 8.
$\bullet$ \textbf{num\_nodes}: Number of nodes in the network. Unless otherwise mentioned, num\_nodes is 256.

{\flushleft \bf Comparing characteristics at Ripple's limits to our system's limits:}
The edge case (current limits or lower bounds) of the existing Ripple system is at $20\%$ malicious nodes. At configurations of $20\%$ malicious nodes SimC, SimRM and SimK are able to come to a right consensus $43.13\%$, $98.6\%$ and $99.53\%$ times out of 1500 simulated cases respectively. SimK is $2.423x$ faster than SimC at achieving at a right consensus at this edge case.

{\flushleft \bf Other Key Results:} Implementing SimK results in lesser time and messages. The complexity of sent SimK messages is a maximum $2x$ SimC and lesser than SimRB. However, considering received messages, SimK's success in significantly lesser received messages than SimC and SimRB. Thus, SimK has higher order of success.
There are an average 5120, 18240 and 18432 links in SimC, SimK and SimRM respectively.

\section{Additional Graphs}

{\flushleft \bf Information Propagation}
Figures \ref{fig:allsimMode2PercentMaliciousComparison}, \ref{fig:allsimMode2SeverityComparison}, \ref{fig:allsimMode46SeverityComparison}, \ref{fig:allsimMode6PercentageEclipsedComparison}, \ref{fig:allsimMode8PercentMaliciousComparison}

{\flushleft \bf Consensus}
Figures \ref{fig:allsimMode1SeverityComparison}, \ref{fig:allsimMode1LatencyComparison}, \ref{fig:allsimMode1LinkLatencyComparison}, \ref{fig:allsimMode1PercentMaliciousComparison}, \ref{fig:allsimMode5SeverityComparison}

\pgfplotsset{scale=0.55}
\begin{figure*}[!htbp]
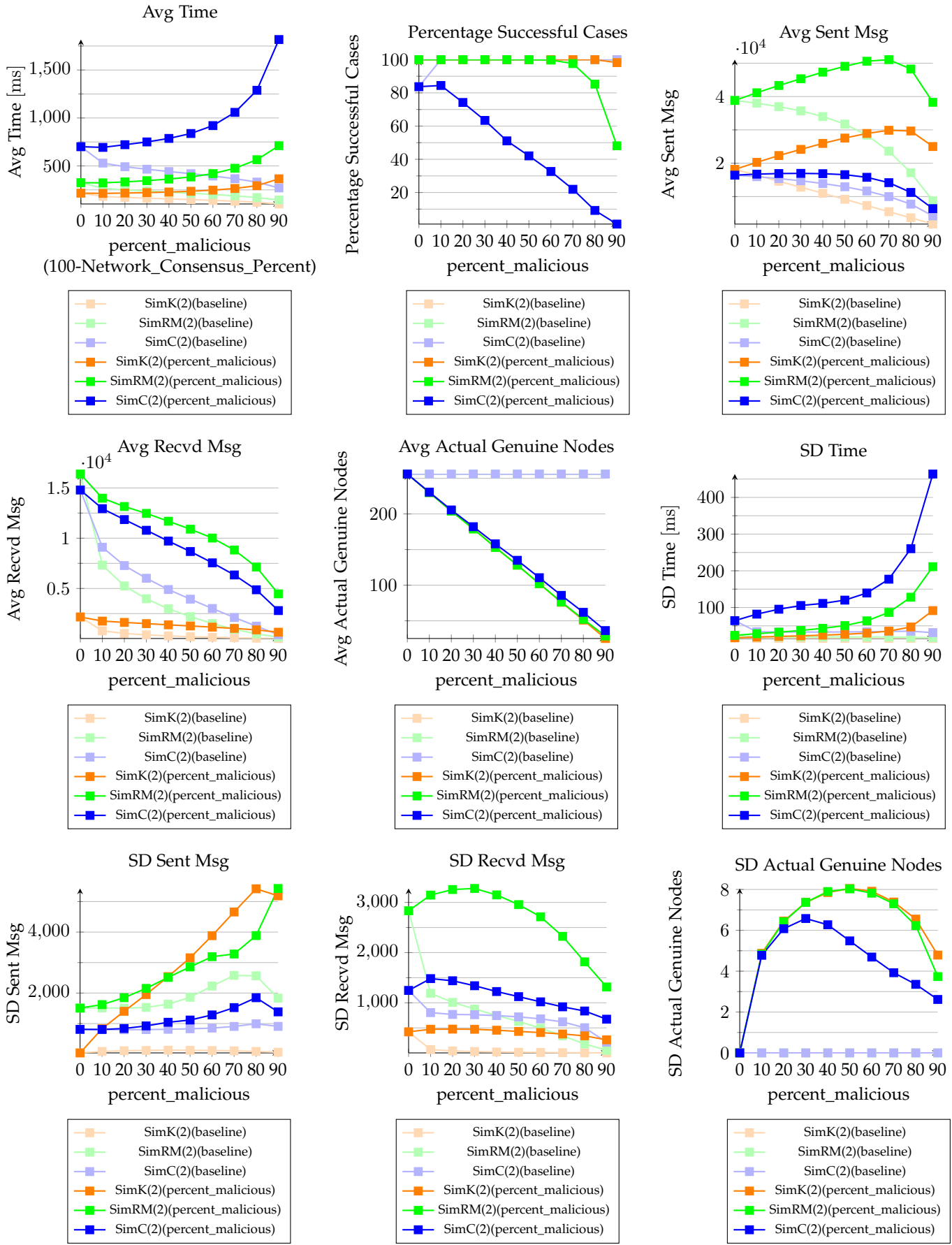

\centering
\subfloat{
\centering
}
\caption{\label{fig:allsimMode2PercentMaliciousComparison}(Information Propagation) Comparison of data points for baseline cases ($percent\_malicious=0$) and percent\_malicious cases in the absence of network issues for mode 2}
\end{figure*}

\pgfplotsset{scale=0.55}
\begin{figure*}[!htbp]
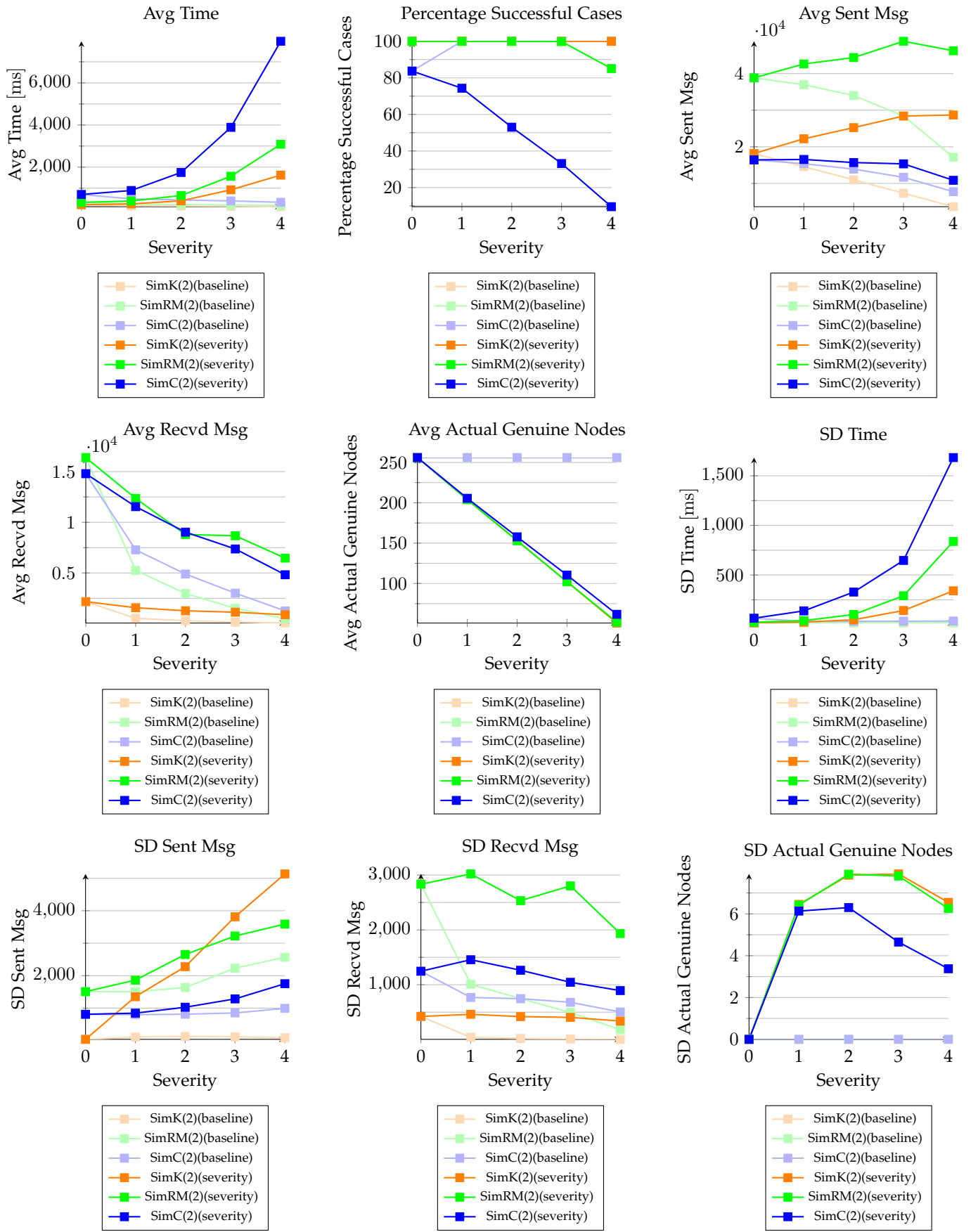

\centering
\subfloat{\centering
%
}
\caption{\label{fig:allsimMode2SeverityComparison}(Information Propagation) Comparison of data points for baseline cases ($percent\_malicious=0$ , Ideal condition, Network\_Consensus\_Percent corresponding to comparison cases' Network\_Consensus\_Percent) and cases of corresponding severity for mode 2}
\end{figure*}

\pgfplotsset{scale=0.55}
\begin{figure*}[!htbp]
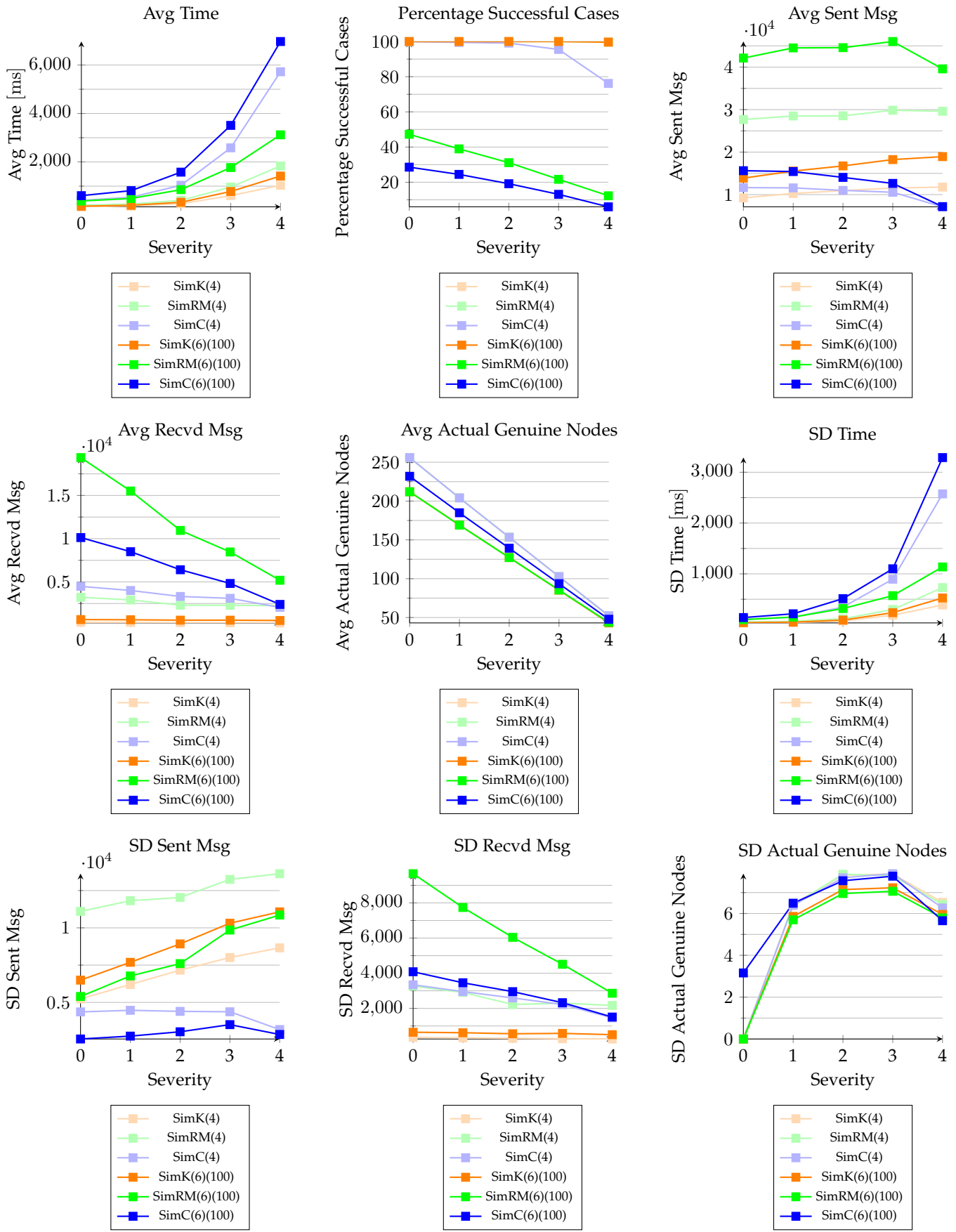

\centering

\subfloat{
\centering
%
}\caption{\label{fig:allsimMode46SeverityComparison}(Information Propagation) Comparison of data points for modes 4 and 6 across varying severity}
\end{figure*}

\pgfplotsset{scale=0.55}
\begin{figure*}[!htbp]
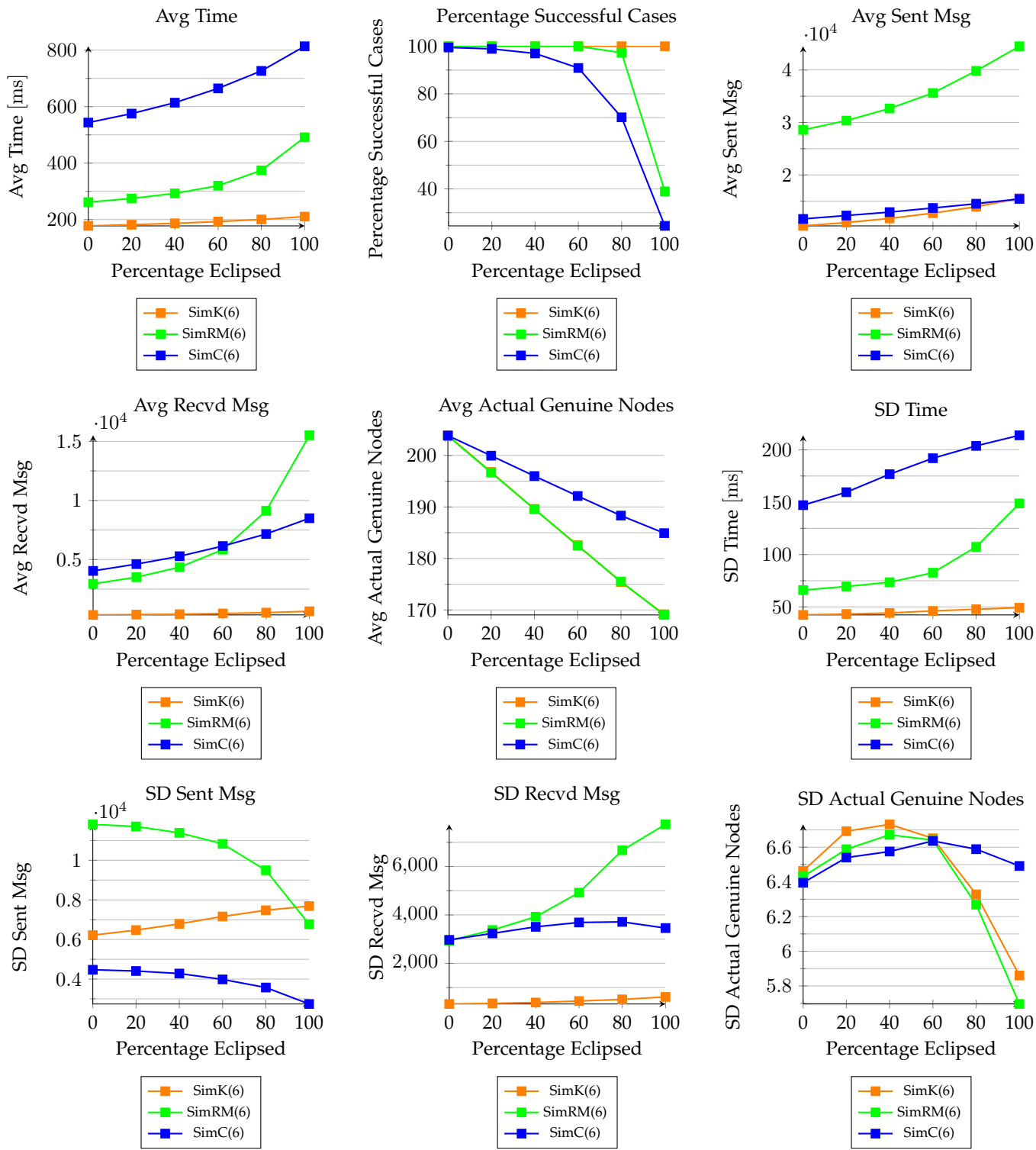

\centering
\subfloat{\centering
%
}\caption{\label{fig:allsimMode6PercentageEclipsedComparison}(Information Propagation) Comparison of data points for mode 6 across varying percentageEclipsed for Real World severity}
\end{figure*}

\pgfplotsset{scale=0.7}
\begin{figure*}[!htb]
\centering
\subfloat{
\centering
%
}\caption{\label{fig:allsimMode8PercentMaliciousComparison}(Information Propagation) Comparison of  data points for mode 7.}
\end{figure*}

\pgfplotsset{scale=0.55}
\begin{figure*}[!htbp]
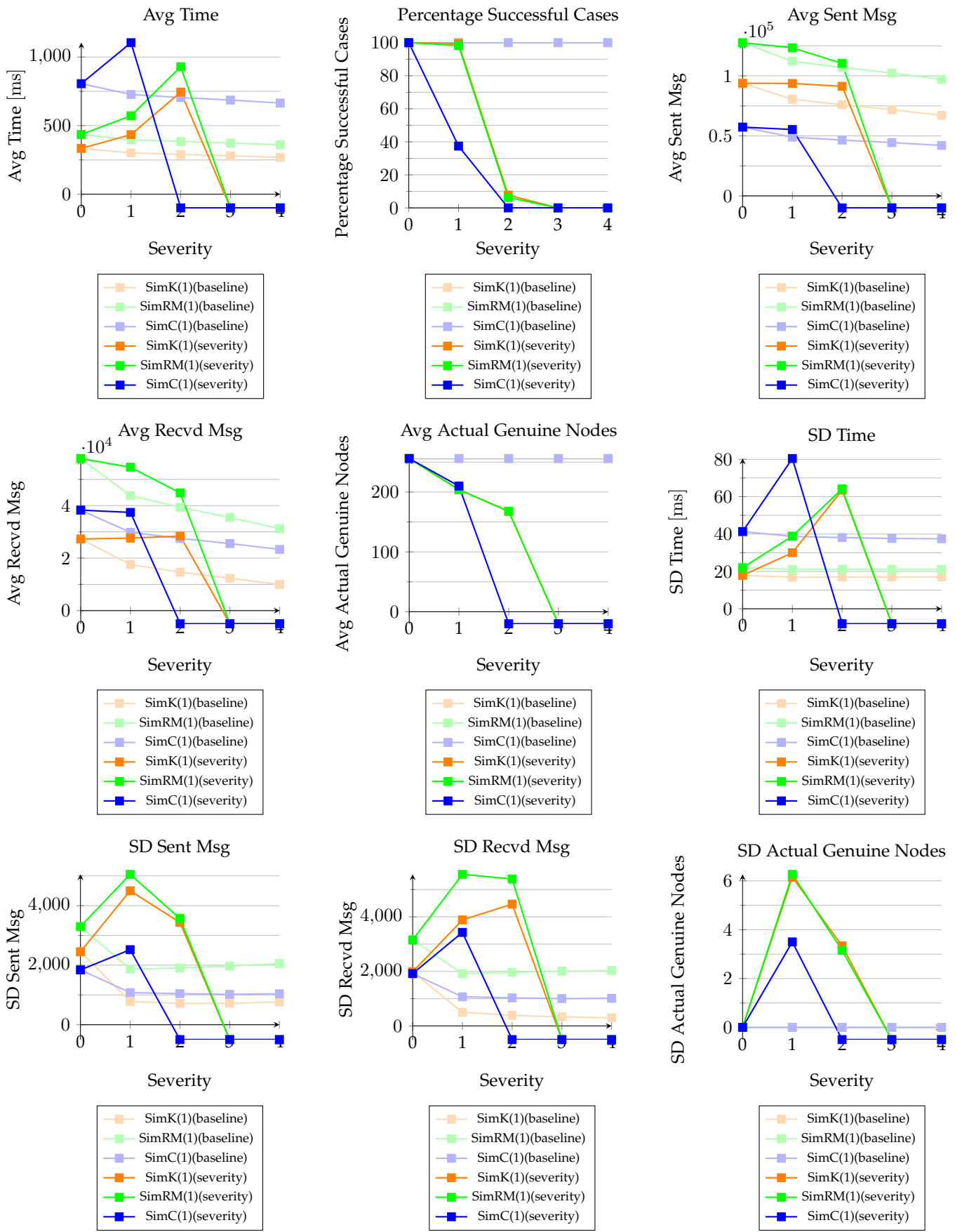

\centering
\subfloat{
\centering
%
}
\caption{\label{fig:allsimMode1SeverityComparison}(Consensus) Comparison of data points for baseline cases ($percent\_malicious=0$ , Ideal condition, Network\_Consensus\_Percent corresponding to comparison cases' Network\_Consensus\_Percent) and cases of corresponding severity for mode 1}
\end{figure*}

\pgfplotsset{scale=0.55}
\begin{figure*}[!htbp]
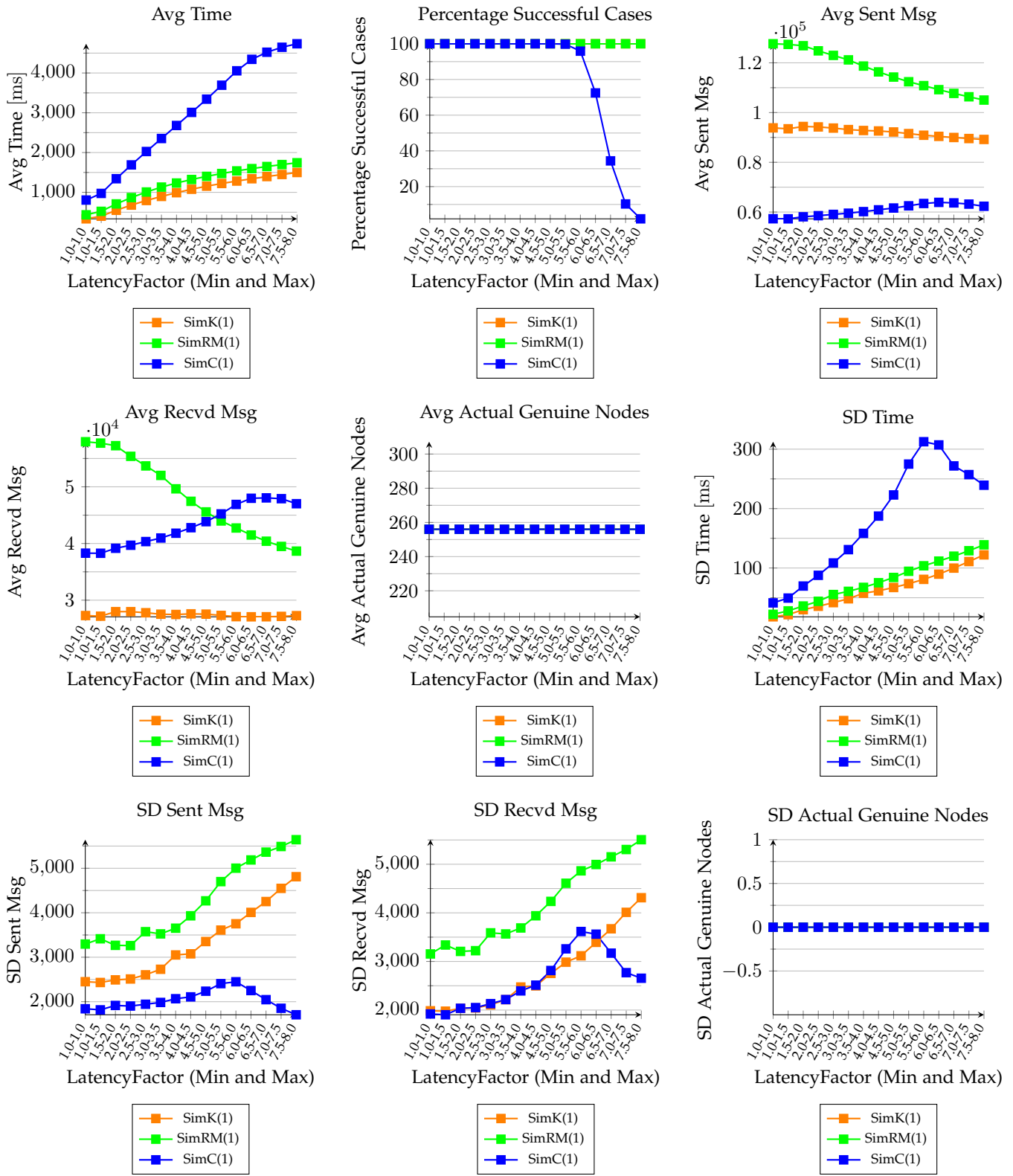

\centering
\subfloat{
\centering
%
}

\caption{\label{fig:allsimMode1LatencyComparison}(Consensus) Comparison of data points for $percent\_malicious=0$ , $percentLinksAffectedByNI=75$, $percentNodesAffectedByNI=100$, $Network\_Consensus\_Percent=100$ and mode 1}
\end{figure*}

\pgfplotsset{scale=0.5}
\begin{figure*}[htbp]
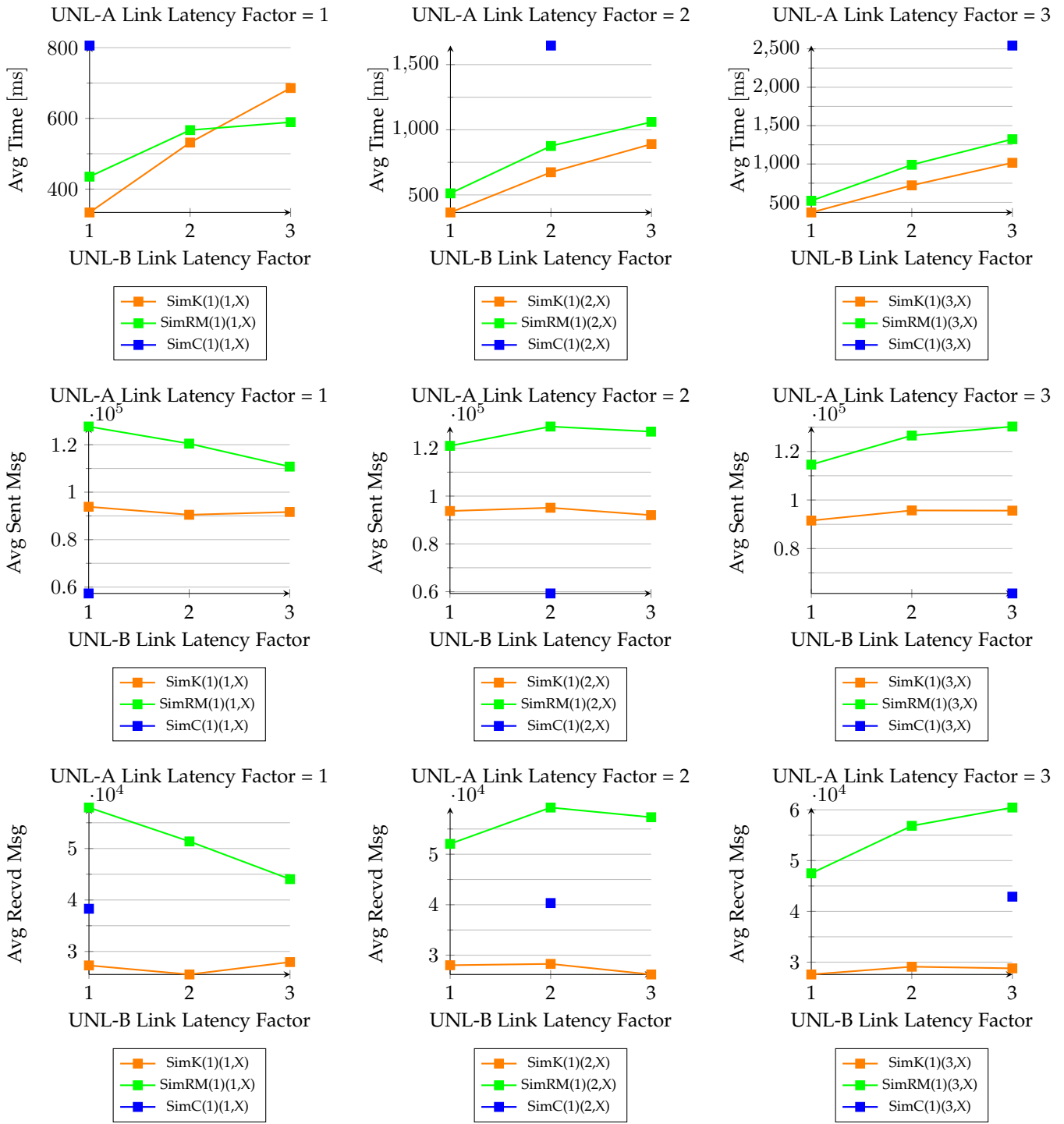

\centering
\subfloat{
\centering
%
}
\caption{\label{fig:allsimMode1LinkLatencyComparison}(Consensus) Comparison of Avg Time, Avg Sent Messages and Avg Received Messages for SimK, SimRM and SimC for varying UNL-A and UNL-B Link Latency Factors (Condition: Ideal Severity, $0\%$ percentage\_malicious, $100\%$ Network\_Consensus\_Percent) for mode 1.}
\end{figure*}

\pgfplotsset{scale=0.55}
\begin{figure*}[!htbp]
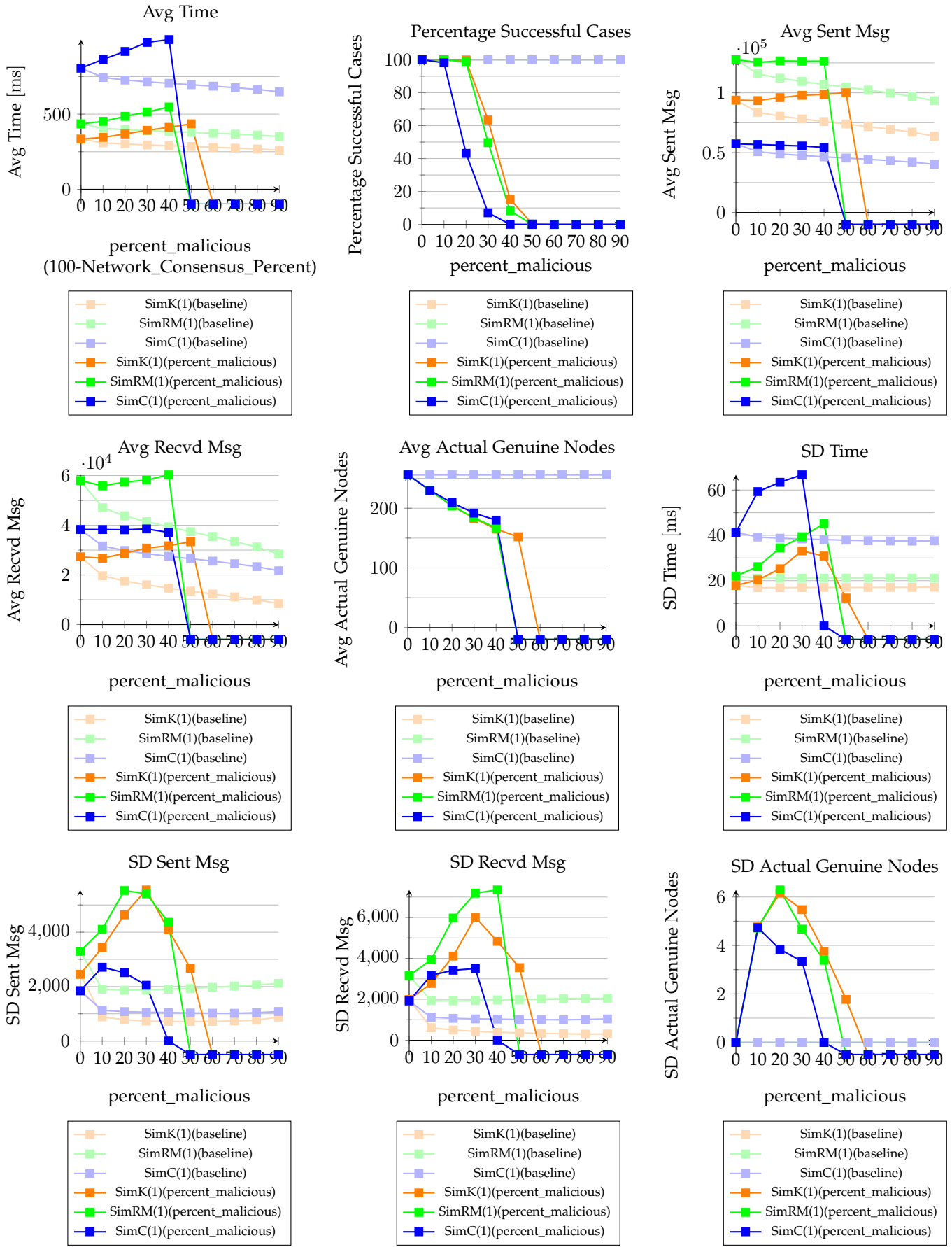

\centering
\subfloat{
\centering
%
}

\caption{\label{fig:allsimMode1PercentMaliciousComparison}(Consensus) Comparison of data points for baseline cases ($percent\_malicious=0$) and percent\_malicious cases in the absence of network issues for mode 1.}
\end{figure*}

\pgfplotsset{scale=0.9}
\begin{figure*}[!tb]
\centering
\subfloat{
\centering
%
}
\caption{\label{fig:allsimMode5SeverityComparison}(Consensus) Comparison of Avg Time and Percent Successful Cases for baseline cases ($percent\_malicious=0$, Ideal condition, Network\_Consensus\_Percent corresponding to comparison cases' Network\_Consensus\_Percent) and cases of corresponding severity for mode 5}
\end{figure*}

\clearpage

\end{document}